\documentclass[conference,onecolumn,11pt]{IEEEtran}
\IEEEoverridecommandlockouts

\usepackage{geometry}  
\usepackage{amsmath,amssymb,amsfonts}

\usepackage[T1]{fontenc}  
\usepackage[utf8]{inputenc}
\usepackage{graphicx}     
\usepackage{url}          
\usepackage{amssymb,amsthm,amsmath}
\usepackage{mathtools}
\usepackage{bbold}
\usepackage{enumitem}
\usepackage{algorithm}
\usepackage{algpseudocode}
\usepackage{todonotes}
\usepackage{booktabs}
\usepackage{multicol}

\usepackage{tikz}
\usetikzlibrary{arrows}
\usetikzlibrary{positioning}
\usetikzlibrary{decorations.pathreplacing}
\usetikzlibrary{fit}
\usetikzlibrary{backgrounds}
\usetikzlibrary{automata}
\usetikzlibrary{positioning}
\usetikzlibrary{arrows}
\usetikzlibrary{calc}
\usetikzlibrary{decorations.markings}
\usetikzlibrary{shapes.geometric}
\usepackage[percent]{overpic}

\definecolor{UmUBlue}{RGB}{42,71,101}
\definecolor{UmUGreen}{RGB}{115,167,144}
\definecolor{UmUGold}{RGB}{215,177,124}
\definecolor{UmUPink}{RGB}{234,186,185}
\definecolor{NoUmUColour}{RGB}{84,142,202}

\newtheorem{theorem}{Theorem}
\newtheorem{lemma}[theorem]{Lemma}
\newtheorem{corollary}[theorem]{Corollary}
\newtheorem{definition}[theorem]{Definition}
\newtheorem{example}[theorem]{Example}

\newcounter{cl}
\newenvironment{claim}{%
  \refstepcounter{cl}%
  \par\medskip\noindent%
  \textit{Claim \thecl.}
}{\par\medskip}

\newcounter{cse}
\newenvironment{firstcase}{%
  \setcounter{cse}{0}%
  \refstepcounter{cse}%
  \par\medskip\noindent%
  \textit{Case \thecse.}
}{\par\medskip}

\newenvironment{case}{%
  \refstepcounter{cse}%
  \par\medskip\noindent%
  \textit{Case \thecse.}
}{\par\medskip}

\newcommand{\nat}{\mathbb N}
\newcommand{\pow}[1]{\wp(#1)}

\newcommand{\emptystr}{\varepsilon}
\newcommand{\ordo}[1]{O(#1)}
\newcommand{\seq}[1]{#1^\circledast}

\newcommand{\funion}{\sqcup}
\newcommand{\types}{\mathcal T}
\newcommand{\sig}{\Sigma}
\newcommand{\trees}[1]{\mathrm T_{#1}}
\newcommand{\addr}[1]{\mathit{addr}(#1)}
\newcommand{\ad}{\alpha}
\newcommand{\alg}{\mathcal A}
\newcommand{\A}{\mathbb A}
\newcommand{\val}{\mathit{val}}

\newcommand{\emptygraph}{\phi}
\newcommand{\labels}{\mathbb L}
\newcommand{\ndlab}{\!\dot{\,\mathbb L}}
\newcommand{\edlab}{\bar{\mathbb L}}
\newcommand{\lab}{\mathit{lab}}
\newcommand{\port}{\mathit{port}}
\newcommand{\dock}{\mathit{dock}}
\newcommand{\src}[1]{\mathit{src}(#1)}
\newcommand{\tar}[1]{\mathit{tar}(#1)}
\newcommand{\union}[2]{\uplus_{#1#2}}
\newcommand{\clone}[2]{\mathit{clone}_{#1}(#2)}
\newcommand{\context}{\mathit{CONT}}
\newcommand{\New}{\mathit{NEW}}
\newcommand{\graphs}{\mathbb G}
\newcommand{\G}{\mathcal G}
\newcommand{\und}[1]{\underline{#1}}
\newcommand{\ext}{\Phi}
\newcommand{\reach}[2]{#1\mathbin{\triangledown} #2}

\newcommand{\prt}[1]{{\mathversion{bold}\scriptsize\color{UmUBlue}$#1$}}
\newcommand{\dk}[1]{{\mathversion{bold}\scriptsize\color{UmUGreen}$(#1)$}}

\newcommand{\profile}{\mathit{profile}}%
\newcommand{\profiles}{\Pi}%
\newcommand{\mult}{\mathit{mult}}%
\newcommand{\greater}[1]{\mathit{gt}_{#1}}%
\newcommand{\equal}[1]{\mathit{eq}_{#1}}%

\theoremstyle{definition}
\newtheorem{myexample}[theorem]{Example}

\begin{document}

\pagestyle{plain}

\title{Polynomial Graph Parsing with\\ Non-Structural Reentrancies
\thanks{This work has been supported in part by the Swedish Research Council under Grant No~2020-03852.}}

\author{
\IEEEauthorblockN{Johanna Bj\"orklund}
\IEEEauthorblockA{
\textit{Dept. Computing Science} \\
\textit{Umeå University}\\
Umeå, Sweden \\
johanna@cs.umu.se}
\and
\IEEEauthorblockN{Frank Drewes}
\IEEEauthorblockA{
\textit{Dept. Computing Science} \\
\textit{Umeå University}\\
Umeå, Sweden \\
drewes@cs.umu.se}
\and
\IEEEauthorblockN{Anna Jonsson}
\IEEEauthorblockA{\textit{Dept. Computing Science} \\
\textit{Umeå University}\\
Umeå, Sweden \\
aj@cs.umu.se}}


\maketitle

\begin{abstract}
Graph-based semantic representations are valuable in natural language processing, where it is often simple and effective to represent linguistic concepts as nodes, and relations as edges between them. Several attempts has been made to find a generative device that is sufficiently powerful to represent languages of semantic graphs, while at the same allowing efficient parsing. We add to this line of work by introducing graph extension grammar, which consists of an algebra over graphs together with a regular tree grammar that generates expressions over the operations of the algebra. 
Due to the design of the operations, these grammars can generate graphs with \emph{non-structural reentrancies}, a type of node-sharing that is excessively common in formalisms such as abstract meaning representation, but for which existing devices offer little support. We provide a parsing algorithm for graph extension grammars, which is proved to be correct and run in polynomial time. 
\end{abstract}

\begin{IEEEkeywords}
Algorithmic Graph Theory, 
Semantic Representations, 
Hyperedge Replacement Grammars, 
Natural Language Processing 
\end{IEEEkeywords}

\section{Introduction}\label{sec:intro}
Semantic representations are a key component in algorithmic solutions for high-level language processing. Word embeddings such as BERT~\cite{Devlin:2019} and GPT-3~\cite{GPT3} have become a popular choice of representation, enabling rapid progress on many tasks. 
Due to the descriptive power of the neural networks that generate them, they rarely submit to analysis by deductive means. 
However, for some applications the transparency gained by formal methods is a must. 
In such cases, we believe that graph-based semantic representations offer a useful complement. 
These generally represent objects as nodes, and relations as directed edges. Probabilities and other weights can be added to reflect quantitative aspects such as likelihoods and uncertainties. 

In this work, we model languages of graph-based semantic representations as the combination of a graph algebra $\alg$ together with a tree grammar $g$ whose generated trees are well-formed expressions over the operations of $\alg$. 
Each tree translates into a graph when evaluated by $\alg$, meaning that the tree language generated by $g$ evaluates to a graph language. 
If we are careful about how we construct and combine $g$ and $\alg$, we can make parsing efficient. 
In other words, given a graph $G$, we can in polynomial time find a tree in the language of $g$ that evaluates to $G$ under $\alg$ (or decide that no such tree exists). 
In the literature, this approach is known as tree-based generation~\cite{drewes:06}. 
The novelty in the present work lies in the design of the  algebra. 

For inspiration, we look at a particular instance of graph-based representations, namely abstract meaning representation (AMR)~\cite{Langkilde:1998,banarescu:2013}. 
It is characterised by its graphs being directed, acyclic, and having unbounded node degree. AMR was first introduced by~\cite{Langkilde:1998} based on a semantic abstraction language by~\cite{Kasper:1989}. 
The notion was refined and popularised by~\cite{banarescu:2013} and instantiated for a limited domain by~\cite{braune:2014}. 
To ground AMR in formal language theory, Chiang et al. \cite{Chiang:2018} analyse the AMR corpus of Banarescu et al~\cite{banarescu:2013}. 
They note that even though the node degree is generally low in practice, this is not always the case, which speaks in favour of models that allow an unbounded node degree. 
Regarding the treewidth of the graphs in the corpus, they find that it never exceeds~4 and conclude that an algorithm can depend exponentially on this parameter and still be feasible in practice.


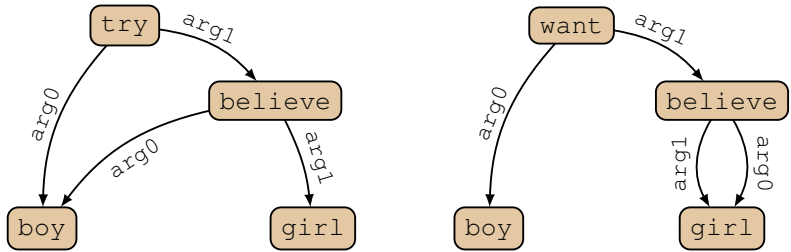
\begin{figure}[t]
    \centering
        \tikzset{
  auto,
  thick, 
  every node/.style={line width=0.25mm, rounded corners, font=\ttfamily\footnotesize,inner sep=2pt, outer sep=0pt}, 
  terminal/.style={line width=0.25mm,rounded corners,draw, fill=UmUGold!70, inner sep=4pt,font=\ttfamily\small},
  every label/.style={draw=none, fill=White, inner sep=0pt, outer sep=0pt}
}
\centering
\begin{tikzpicture}
   \node[terminal]  (b)  at (-.6,-1.5) {boy}; 
   \node[terminal]  (w1) at (.5,1.2) {try}; 
   \node[terminal]  (b1) at (2.5,.2) {believe};
   \node[terminal]  (g)  at (3,-1.5) {girl};
    \path[-latex]
    (w1) edge [line width=0.25mm,-latex,bend right=20,pos=0.5, above, swap,sloped]     node {arg0} (b)
         edge [line width=0.25mm,-latex,bend left=20,pos=0.4,sloped]           node {arg1} (b1)
    (b1) edge [line width=0.25mm,-latex,swap, bend right=20,sloped] node[outer sep=1pt] {arg0} (b)
         edge [line width=0.25mm,-latex,bend left=10,sloped]      node[outer sep=1pt] {arg1} (g);
\end{tikzpicture}
\qquad\quad
\begin{tikzpicture}
   \node[terminal]  (b)  at (-.6,-1.5) {boy}; 
   \node[terminal]  (w1) at (.5,1.2) {want}; 
   \node[terminal]  (b1) at (2.5,.2) {believe};
   \node[terminal]  (g)  at (2.5,-1.5) {girl};
    \path[-latex]
    (w1) edge [line width=0.25mm,-latex,bend right=20,pos=0.5, above, swap,sloped]     node {arg0} (b)
         edge [line width=0.25mm,-latex,bend left=20,pos=0.4,sloped]           node {arg1} (b1)
    (b1) edge [line width=0.25mm,-latex,swap, bend left=30, above, sloped] node[outer sep=1pt] {arg0} (g)
         edge [line width=0.25mm,-latex,bend right=30,sloped]      node[outer sep=1pt] {arg1} (g);
\end{tikzpicture}
        \caption{AMR graphs for the pair of sentences  ``The boy tries to believe the girl'' and ``The boy wants the girl to believe herself''. The nodes represent the concepts `boy', `girl', `try', `want', and `believe'. The edges labelled `arg0' give the agent for each verb, and the edges labelled `arg0' the patient. }
    \label{fig:reent}
\end{figure}

In the context of semantic graphs, it is common to talk about \emph{reentrancies}. 
Figure~\ref{fig:reent} illustrates this concept with a pair of AMR graphs, both of which require node sharing, or \emph{reentrant edges}, to express the correct semantics. 
We propose to distinguish between two types of reentrancies: 
In the first graph, the reentrancy is \emph{structural}, meaning that the boy must be the agent of whatever he is trying, so the `arg0' edge can only point to him. 
In the second graph, the reentrancy is \emph{non-structural} in the sense that although in this particular case the girl believes in herself, there is nothing to prevent her believing in someone else, or that someone else believes in her.  
In general, we speak of (i) structural reentrancies when they are syntactically governed, e.g.,~by control or coordination, and of (ii) non-structural reentrancies when they represent coreferences which can in principle refer to any antecedent, essentially disregarding the structure of the graph.  
Another example of (i) is subject control as in ``They persuaded him to talk to her'', where the person who talks and the person who was persuaded must be one and the same. 
In contrast, ``[\ldots], but she liked them'' is an example of (ii) since antecedent of ``them'' may be picked from anywhere in ``[\ldots]''  for the semantic representation to be  valid.

Another important characteristic of AMR and other notions of semantic graphs in natural language processing is that certain types of edges, like those in Figure~\ref{fig:reent}, point to \emph{arguments} of which there should be at most one. 
For example, the concept `try' in Figure~\ref{fig:reent} must only have one outgoing `arg0' edge. 
Other edges, of types not present in Figure~\ref{fig:reent}, can for example represent modifiers, such as believing \emph{strongly} and \emph{with a passion}. 
Outgoing edges of these types are usually not limited in number.

We have previously shown that contextual hyperedge-replacement grammars (CHRGs), an  extension of the well-known context-free hyperedge-replacement grammars (HRGs), can handle both structural and non-structural reentrancies~\cite{drewes:2017}. 
Moreover, a parser generator has been provided by Drewes, Hoffmann, and Minas for CHRGs that, when it succeeds in producing a parser, guarantees that the parser will run in quadratic (and in the common case, linear) time in the size of its input graph~\cite{Drewes-Hoffmann-Minas:19,Drewes-Hoffmann-Minas:21}. 
However, similarly to common LL- and LR-parsers for context-free languages the parser generator may discover a parsing conflict, thus failing to output a parser. 
Compared to the string case, the possible reasons for conflicts are much subtler, which makes grammar construction a complex and error-prone endeavour.

For this reason, we are now proposing a CHRG variant that by construction allows for polynomial parsing while retaining the ability to specify graph languages with reentrancies of both type~(i) and~(ii). 
One key property of these grammars that enables polynomial parsing is that, intuitively, nodes are provided with all of their outgoing edges the moment they are created. 
Using ordinary contextual hyperedge replacement, this would thus result in graph languages of bounded out-degree, the bound being given by the maximal number of outgoing edges of the nodes in the right-hand sides of production rules. 
To be able to generate semantic graphs in which a concept can have an arbitrary number of edges representing modifiers, we use a technique from~\cite{DrewesEtAl:2010} known as \emph{cloning}. 
To incorporate both contextuality and cloning in a natural way without sacrificing efficient parsability, we formalise our grammars using the concept of \emph{tree-based graph generation}. 
While tree-based hyperedge-replacement grammars are well-known to be equivalent to ordinary ones (see below), the tree-based formulation does make a difference in the contextual case as it avoids the problem of cyclic dependencies that has not yet been fully characterised and can make parsing intractable~\cite{DrewesHoffmannMinas:2019psrjournal}.

\subsection{Related work}
Tree-based generation dates back to the seminal paper by Mezei and Wright~\cite{Mezei-Wright:67} which generalises context-free languages to languages over arbitrary domains, by evaluating the trees of a regular tree language with respect to an algebra.\footnote{Mezei and Wright formulate this in terms of systems of equations, but the essential ideas are the same.}
Operations on graphs were used in this way for the first time by Bauderon and Courcelle~\cite{Bauderon:1987}. See the textbook~\cite{Courcelle-Engelfriet:12} by Courcelle and Engelfriet for the eventual culmination of this line of work. 
Graph operations similar to those used here (though without the contextual extension) appeared first in Courcelle's work on the monadic-second order logic of graphs~\cite[Definition~1.7]{Courcelle:91b}. 
Essentially, if the right-hand side of a production contains $k$ nonterminal hyperedges, it is viewed as an operation that takes $k$ hypergraphs as arguments (corresponding to the hypergraphs generated by the $k$~subderivations) and returns the hypergraph obtained by replacing the nonterminals in the right-hand side by those $k$ argument hypergraphs. 
By context-freeness, evaluating the tree language that corresponds to the set of derivation trees of the grammar (which is a regular tree language) yields the same set of hypergraphs as is generated by the original HRG.

Several formalisms have been put forth in the literature to describe graph-based semantic representations in general and AMR in particular. 
Most of these can be seen as variations of HRGs; see \cite{HabelKreowski:1987,Bauderon:1987,drewes:1997}. 
It was established early that unrestricted HRGs can generate NP-complete graph languages~\cite{aalbersberg:1986,lange:1987}, so restrictions are needed to ensure efficient parsing. 
To this end, Lautemann~\cite{Lautemann:1990} proposed a CYK-like membership algorithm and proved that it runs in polynomial time provided that the language satisfies the following condition: for every graph $G$ in the language, the number of connected components obtained by removing $s$ nodes from $G$ is in $\ordo{\log n}$, where $n$ is the number of nodes of $G$ and the constant $s$ depends on the grammar. 
Lautemann's algorithm was refined by Chiang et al.~\cite{chiang:2013} to make it more suitable for NLP tasks, but the algorithm is exponential in the node degree of the input graph.


In a parallel line of work, Quernheim and Knight \cite{quernheim-knight-2012-towards} define automata for directed acyclic graphs (DAGs) to process feature structures in machine translation. 
Chiang et al. \cite{Chiang:2018} propose an extended model of DAG automata, focusing on semantic representations such as AMR. 
Towards this end, they extend DAG automata by allowing the left- and right-hand sides to have the form of a restricted regular expression.  
Blum and Drewes~\cite{Blum-Drewes:17} complement this work by studying  language-theoretic properties of the DAG automata. 
They establish, among other things, that equivalence and emptiness are decidable in polynomial time. 

Various types of hyperedge replacement graph algebras for AMR parsing are described in the work by Groschwitz et al., see, e.g. \cite{Groschwitz-Koller-Teichmann:15,groschwitz-etal-2017-constrained,Groschwitz.etAl:18,Lindemann.etAl:19,Lindemann.etAl:20}. 
A central objective is to find linguistically motivated restrictions that can efficiently be trained from data. 
An algorithm based on such an algebra for translating strings into semantic graphs was presented in~\cite{Groschwitz.etAl:18}: operations of arity zero denote graph fragments, and operations of arity two denote binary combinations  of graph fragments into larger graphs. 
The trees over the operations of this algebra mirror the compositional structure of semantic graphs. 
The approach differs from ours in that the graph operations are entirely deterministic, and that neither context nodes nor cloning are used.
Moreover, as is common in computational linguistics, evaluation is primarily empirical.

Through another set of syntactic restrictions on HRGs, Björklund et al.~\cite{Bjorklund.etAl:21b,Bjorklund:2019b} arrive at order-preserving dag grammars (OPDGs) -- with or without weights -- which can be parsed in linear time. 
Intuitively, the restrictions combine to ensure that each generated graph can be uniquely represented by a tree in a particular graph algebra. 
Despite their restrictions, OPDGs can describe central structural properties of AMR, but their limitation lies in the modelling of reentrancies. 
Out of the previously discussed types of reentrancies, type (i) can to a large extent be modelled using OPDGs. 
Modelling type (ii) cannot be done (except in very limited cases) since it requires attaching edges to the non-local context in an stochastic way, which cannot be achieved using hyperedge replacement alone.

CHRGs~\cite{drewes:2012contextual} extend the ordinary HRGs with so-called \emph{contextual} rules, which allow for isolated nodes in their left-hand-sides.
Contextual rules can reference previously generated nodes, that are not attached to the replaced nonterminal hyperedge, and add  structure to them. 
Even though this formalism is strictly more powerful than HRG, it inherits several of the nice properties of HRG. 
In particular, there are useful normal forms and the membership problem is in NP~\cite{drewes:2015contextual}. 
Paralleling previous work on HRG, we set out to find a syntactically restricted subclass of CHRG, enriched by a benign form of cloning, that is sufficiently descriptive to model AMR, while having a membership problem in P.

\section{Graphs and Graph Extension Grammars}
We first recall some standard definitions and notations from discrete mathematics, automata theory, and algorithmic graph theory. 
The set of natural numbers (including $0$) is denoted by $\nat$, and for $n\in\nat$, $[n]$ denotes $\{1,\dots,n\}$. 
In particular, $[0]=\emptyset$. 
For a set $S$, the set of all finite sequences over $S$ is denoted by $S^*$; the subset of $S^*$ containing only those sequences in which no element of $S$ occurs twice is denoted by $\seq S$. 
Both contain the empty sequence $\emptystr$. For $w\in S^*$, $[w]$ denotes the set of all elements of $S$ occurring in $w$.\footnote{The similarity to the notation $[n]$, $n \in \nat$, is intentional, since both operations turn single elements into sets.} 
The canonical extensions of a function $f\colon S\to T$ to $S^*$ and to $\pow S$ are denoted by $f$ as well, i.e., $f(s_1\cdots s_n)=f(s_1)\cdots f(s_n)$ for $s_1,\dots,s_n\in S$ and $f(S')=\{f(s)\mid s\in S'\}$ for all $S'\subseteq S$. 
The restriction of $f\colon S\to T$ to $S'\subseteq S$ is denoted by $f|_{S'}$. 
The power set of $S$ is denoted by $\pow{S}$. For functions $f_1\colon S_1\to T_1$ and $f_2\colon S_2\to T_2$ by $f_1\cup f_2$, we let $f_1\funion f_2\colon S_1\cup S_2\to T_1\cup T_2$ be given by
\[
(f_1\funion f_2)(s)=\left\{\begin{array}{@{}ll@{}}
f_1(s)&\text{if $s\in S_1$}\\
f_2(s)&\text{if $s\in S_2\setminus S_1$\enspace.}
\end{array}\right.
\]

\subsection*{Trees and Algebras}
As sketched in the introduction, we generate trees over an algebra $\alg$ of graph operations and then evaluate each tree $t$ in a non-deterministic fashion to a set of graphs $\val_\alg(t)$. 
To be able to ensure that the graph operations are always well defined on their arguments, we work with \emph{typed} trees. 
%
Let $\types$ be a set of \emph{types}. 
A \emph{$\types$-typed signature} is a pair $A=(\sig,\sigma)$ such that $\sig$ is a finite set of symbols and $\sigma$ is a function that assigns to every $f\in\sig$ a pair $\sigma(f)\in \types^*\times \types$. 
We usually keep $\sigma$ implicit, thus identifying $A$ with $\sig$, and write $f\colon \tau_1\cdots \tau_k\to \tau$ to indicate that $\sigma(f)=(\tau_1\cdots \tau_k,\tau)$. 
The number $k$ is said to be the \emph{rank} of $f$ in $\sig$. 
Furthermore, we may drop the epithet $\types$-typed when speaking of signatures in case $\types$ is of little relevance or understood from context.

The family $\trees\sig=(\trees\sig^\tau)_{\tau\in \types}$ of all well-formed trees over a signature $\sig$ is defined inductively, as usual. 
It is the smallest family of sets $\trees\sig^\tau$ of formal expressions such that, for every $f\colon \tau_1\cdots \tau_k\to \tau$ and all $t_1\in\trees\sig^{\tau_1},\dots,t_k\in\trees\sig^{\tau_k}$, we have $f[t_1,\dots,t_k]\in\trees\sig^\tau$. 
Note that, for every tree $t\in\bigcup_{\tau\in \types}\trees\sig^\tau$ there is a unique type $\tau$ such that $t\in\trees\sig^\tau$. 
Therefore, we may simply write $t\in\trees\sig$ and call $\tau$ the type of $t$.

We refer to the subtrees of a tree by their Gorn addresses. 
Thus, we define the set $\addr t\subseteq\nat^*$ of \emph{addresses} in a tree $t$, and the subtrees they designate, in the usual way: for $t=f[t_1,\dots,t_k]$, $\addr t=\{\emptystr\}\cup\bigcup_{i\in[k]}\{i\ad\mid \ad\in\addr{t_i}\}$ (note that the basis of the recursion is the case $k=0$).
Moreover, for all $\ad\in\addr t$
\[
t/\ad=\left\{\begin{array}{@{}ll@{}}
t & \text{if $\ad=\emptystr$}\\
t_i/\ad' & \text{if $\ad=i\ad'$ for some $i\in[k]$ and $\ad'\in\addr{t_i}$\enspace.}
\end{array}\right.
\]

Let $\sig$ be a $\types$-typed signature. 
A \emph{$\sig$-algebra} is a pair $\alg=((\A_\tau)_{\tau\in \types},(f_\alg)_{f\in \sig})$ where $\A_\tau$ is a set for each $\tau\in \types$ and $f_\alg$ is a function $f_\alg\colon \A_{\tau_1}\times\cdots\times\A_{\tau_k}\to \A_\tau$ for every $f\colon \tau_1\cdots \tau_k\to \tau$ in $\sig$. 
Given a tree $t\in\trees\sig$, the result of evaluating $t$ with respect to $\alg$ is denoted by $\val_\alg(t)$. 
It is defined recursively: if $t=f[t_1,\dots,t_k]$ then
\[
\val_\alg(t)=f_\alg(\val_\alg(t_1),\dots,\val_\alg(t_k))\enspace.
\]

\subsection*{Regular Tree Languages}

To generate trees (i.e., expressions) over the operations of an algebra, we use \emph{regular tree grammars}. 
Our interest is in algebras over graph domains, but to illustrate the general idea, we can look at an algebra over the rational numbers $\mathbb{Q}$. 
In this setting, the operations could be the standard arithmetic ones, such as $+, -, /, \times$, all binary, together with constant-valued operations in $\mathbb{Q}$, all of arity zero. 
It is hopefully easy to see how the trees $+[\times[2, \frac{1}{2}],2]$ and $/[10,-[10,4]]$ evaluate to $2$ and $\frac{5}{3}$, respectively. 
Moreover, by introducing types, we may work with an algebra over a heterogeneous domain, e.g., natural numbers and rationals, and use the types to ensure that when a tree is evaluated, operations are always applied to suitable arguments.

\begin{definition}[regular tree grammar]
Let $\types$ be a set of types. A $\types$-typed \emph{regular tree grammar} (or, for brevity, regular tree grammar) is a system $g=(N,\sig,P,S)$ consisting of the following components:
\begin{itemize}
    \item $N$ is a $\types$-typed signature of symbols of rank~$0$ called \emph{nonterminals}; we write $N_\tau$ to denote the set of all nonterminals $A\in N$ such that $A\colon\emptystr\to\tau$, where $\tau\in\types$.
    \item $\sig$ is a $\types$-typed signature of \emph{terminals} which is disjoint with $N$.
    \item $P$ is a finite set of \emph{productions} of the form $A\to f[A_1,\dots,A_k]$ where, for some types $\tau,\tau_1,\dots,\tau_k\in \types$, we have $f\colon \tau_1\cdots \tau_k\to \tau$, $A_i\in N_{\tau_i}$ for all $i\in[k]$, and $A\in N_\tau$.
    \item $S\in N_\tau$ is the \emph{initial nonterminal} for some type $\tau\in \types$.
\end{itemize}
We also say that $g$ is a regular tree grammar \emph{over $\sig$}.
\end{definition}

\begin{definition}[regular tree language]
Let $g=(N,\sig,P,S)$ be a regular tree grammar. The family $(L_A(g))_{A\in N}$ is the smallest family of sets of trees such that, for all $A\in N$, $L_A(g)$ is the set of all trees $f[t_1,\dots,t_k]$  such that there is a production $(A \to f[A_1,\dots,A_k])\in P$ for some $k\in\nat$ and $A,A_1,\dots,A_k\in N$, and $t_i\in L_{A_i}(g)$ for all $i\in[k]$.
The \emph{regular tree language} generated by $g$ is $L(g)=L_S(g)$.
\end{definition}

Note that, as can be shown by a straightforward induction, $L_A(g)\subseteq\trees\sig^\tau$ for all $A\in N_\tau$. In particular $L(g)\subseteq\trees\sig^\tau$ where $\tau$ is the type of $S$, i.e., $S\in N_\tau$.

\subsection*{Graphs and Graph Operations}

We now define the type of graphs used in this paper, and the operations used to construct them. 
In short, we work with node- and edge-labelled directed graphs, each equipped with a finite set of ports. From a graph operation point of view, the sequence of ports is the ``interface'' of the graph; its nodes are the only ones that can individually be accessed. 
To ensure that the resulting graph is well-formed, we need to pay attention to the number of ports, which determines the \emph{type} of the graph.

\begin{definition}[graph]

A \emph{labelling alphabet} (or simply \emph{alphabet}) is a pair $\labels=(\ndlab,\edlab)$ of finite sets $\ndlab$ and $\edlab$ of labels. A \emph{graph} (over $\labels$) is a system $G=(V,E,\lab,\port)$ where
\begin{itemize}
    \item $V$ is a finite set of \emph{nodes},
    \item $E\subseteq V\times\edlab\times V$ is a (necessarily finite) set of \emph{edges},
    \item $\lab\colon V\to\ndlab$ assigns a node label to every node, and
    \item $\port\in\seq V$ is a sequence of nodes called \emph{ports}.
\end{itemize}
The \emph{type} of $G$ is $|\port|$. The set of all graphs (over an implicitly understood labelling alphabet) of type $\tau$ is denoted by $\graphs_\tau$. For an edge $e=(u,\ell,v)\in E$, we let $\src e=u$ and $\tar e=v$.
\end{definition}

An \emph{isomorphism} (on $G$) is a bijective function $\mu\colon V\to V'$ for some set $V'$ of nodes. 
It maps $G$ to the graph $\mu(G)=(V',E',\mu\circ\lab,\mu(\port))$ where $E'=\{(\mu(u),\ell,\mu(v))\mid (u,\ell,v)\in E\}$. 
If such an isomorphism exists, then the graphs $G$ and $\mu(G)$ are said to be \emph{isomorphic}.

The graph operations are of two kinds. The first is essentially disjoint union where the port sequences of the argument graphs are concatenated.
Let $\tau,\tau'\in\nat$. The operation $\union\tau{\tau'}\colon\graphs_\tau\times\graphs_{\tau'}\to\graphs_{\tau+\tau'}$ is defined as follows: for $G\in\graphs_\tau$ and $G'\in\graphs_{\tau'}$, $\union\tau{\tau'}(G,G')$ yields the graph in $\graphs_{\tau+\tau'}$ obtained by making the two graphs disjoint through a suitable renaming of nodes and taking their union. This is defined in the obvious way for the first three components and concatenates the port sequences of both graphs. Thus, the union operation $\union\tau{\tau'}$ is not commutative and only defined up to isomorphism. We usually write $G\union\tau{\tau'}G'$ instead of $\union\tau{\tau'}(G,G')$. To avoid unnecessary technicalities, we shall generally assume that $G$ and $G'$ are disjoint from the start, and that no renaming of nodes takes place. We extend $\union\tau{\tau'}$ to an operation $\union\tau{\tau'}\colon\pow{\graphs_\tau}\times\pow{\graphs_{\tau'}}\to\pow{\graphs_{\tau+\tau'}}$ in the usual way: for $\G\subseteq\graphs_\tau$ and $\G'\in\graphs_{\tau'}$, $\G\union\tau{\tau'}\G'=\{G\union\tau{\tau'}G'\mid G\in\G_\tau,\ G'\in\G'\}$.

To introduce the second type of graph operation, we first define an auxiliary nondeterministic \emph{cloning operation}. The purpose of this operation is to make the expansion operations (to be defined afterwards) more flexible by allowing parts of it to be replicated an arbitrary number of times. This increases the flexibility with which it can  attach to an existing structure. The operation was originally introduced to formalise the structure of object-oriented programs~\cite{DrewesEtAl:2010}, and  later adopted in computational linguistics~\cite{Bjorklund:2019b}.

\begin{figure}[t]
    \centering
        \tikzset{
  every label/.style={draw=none, fill=none, inner sep=1pt}
}

     \begin{tikzpicture}[xscale=.5, yscale=.75]
     \tikzstyle{overtex}=[circle,line width=0.30mm,fill=white,draw,inner sep=.5pt, minimum size=2ex];

    \node at (1,1.5) [overtex] (nx1) {\footnotesize $b$};
     \node at (4.75,1.5) [overtex] (nx2) {\footnotesize $a$};
     \node at (7,1.5) [overtex] (nx3) {\footnotesize $b$};
     
     \node at (0,0) [overtex] (ny1) {\footnotesize $c$};
     \node at (2,0) [overtex, fill=UmUGreen!30, label=below:{\Large {\textcolor{UmUGreen}{$\star$}}}] (ny2) {\footnotesize $c$};
     \node at (4,0) [overtex] (ny3) {\footnotesize $b$};
     \node at (6,0) [overtex] (ny4) {\footnotesize $b$};
     \node at (8,0) [overtex] (ny5) {\footnotesize $a$};
     
     
    \draw[->,line width=0.30mm,-latex,bend left=7] (nx1) to (nx2);
    \draw[->,line width=0.30mm,-latex,bend right=10] (nx1) to (ny1);
    \draw[->,line width=0.30mm,-latex,bend left=10] (nx1) to  (ny2) ;
    \draw[->,line width=0.30mm,-latex,bend right=10] (nx2) to (ny2);
    \draw[->,line width=0.30mm,-latex,bend right=5] (nx2) to (ny3);
    \draw[->,line width=0.30mm,-latex,bend left=10] (nx2) to (ny5);      
          
      \begin{scope}[shift={(11.5,0)}]      
     \node at (1,1.5) [overtex] (x1) {\footnotesize $b$};
     \node at (10.75,1.5) [overtex] (x2) {\footnotesize $a$};
     \node at (13,1.5) [overtex] (x3) {\footnotesize $b$};
     
     \node at (0,0) [overtex] (y1) {\footnotesize $c$};
     \node at (2,0) [overtex, fill=UmUGreen!30] (y2a) {\footnotesize $c$};
     \node at (4,0) [overtex, fill=UmUGreen!30] (y2b) {\footnotesize $c$};
     \node at (6,0)  (y2c) {...};
     \node at (8,0) [overtex, fill=UmUGreen!30] (y2d) {\footnotesize $c$};
     \node at (10,0) [overtex] (y3) {\footnotesize $b$};
     \node at (12,0) [overtex] (y4) {\footnotesize $b$};
     \node at (14,0) [overtex] (y5) {\footnotesize $a$};
     
     
    \draw[->,line width=0.30mm,-latex,bend left=7] (x1) to (x2);
    \draw[->,line width=0.30mm,-latex,bend right=10] (x1) to (y1);
    \draw[->,line width=0.30mm,-latex,bend left=10] (x1) to  (y2a);
    \draw[->,line width=0.30mm,-latex,bend right=10] (x2) to (y2a);
    \draw[->,line width=0.30mm,-latex,bend left=10] (x1) to  (y2b) ;
    \draw[->,line width=0.30mm,-latex,bend right=10] (x2) to (y2b);
    \draw[->,line width=0.30mm,-latex,bend left=10] (x1) to  (y2d) ;
    \draw[->,line width=0.30mm,-latex,bend right=10] (x2) to (y2d);
    \draw[->,line width=0.30mm,-latex,bend right=5] (x2) to (y3);
    \draw[->,line width=0.30mm,-latex,bend left=10] (x2) to (y5);
    \end{scope}
    \end{tikzpicture}
    
        \caption{The cloning operation takes a graph  $G$ and a subset $C$ of its nodes, and replaces each of the nodes in $C$ and their incidents edges by an arbitrary number of copies. The above figure shows $G$ on the left-hand side, with $C$ containing the single node marked with a star. On the right-hand side, the node and attaching structure has been replaced by  a number of clones.}
    \label{fig:cloning}
\end{figure}
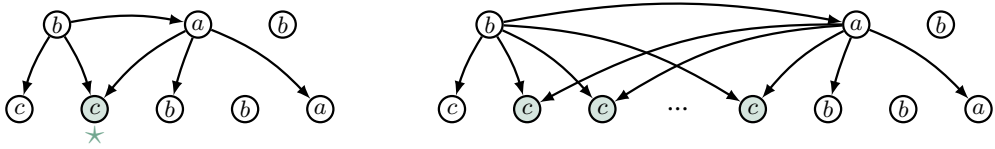

To define the cloning operation, consider a graph $G=(V,E,\lab,\port)$ and let $C\subseteq V\setminus[\port]$. Then $\clone CG$ is the set of all graphs obtained from $G$ by replacing each of the nodes in $C$ and their incident edges by an arbitrary number of copies (Figure~\ref{fig:cloning} illustrates the construction). Formally, $G'=(V',E',\lab',\port')\in\clone CG$ if there is a family $(C_v)_{v\in V}$ of pairwise disjoint sets $C_v$ of nodes, such that the following hold:
\begin{itemize}
    \item $C_v=\{v\}$ for all $v\in V\setminus C$,
    \item $V'=\bigcup_{v\in V}C_v$,
    \item $E'=\bigcup_{(u,\ell,v)\in E}C_u\times\{\ell\}\times C_v$, and
    \item for all $v\in V$ and $v'\in C_v$, $\lab'(v')=\lab(v)$.
\end{itemize}
Note that cloning does \emph{not} rename nodes in $V\setminus C$. In the following, we shall continue to denote by $C_v$ ($v\in V$) the set of clones of $v$ in a graph in $\clone CG$.

We can now define the second type of operation. It is a unary operation called \emph{graph expansion operation} or simply \emph{expansion operation}. We will actually use a restricted form, called \emph{extension operation}, but define the more general expansion operation first. An expansion operation is described by a graph enriched by two additional components: let
\[
\ext=(V_\ext,E_\ext,\lab_\ext,\port_\ext,\dock_\ext,C_\ext)
\]
where both $\und\ext=(V_\ext,E_\ext,\lab_\ext,\port_\ext)$, henceforth called the \emph{underlying graph} of $\ext$, and $(V_\ext,E_\ext,\lab_\ext,\dock_\ext)$ are graphs, and $C_\ext\subseteq V_\ext\setminus([\port_\ext]\cup[\dock_\ext])$. 
We call $\context_\ext=V_\ext\setminus([\port_\ext]\cup[\dock_\ext])$ the set of \emph{context nodes} and $\New_\ext=[\port_\ext]\setminus[\dock_\ext]$ the \emph{new nodes}. 
The nodes in $C_\ext$ are the \emph{clonable nodes}.
This terminology reflects that, according to the following definition, an application of $\ext$ to a graph adds new (copies of the) nodes in $\New_\ext$ while those in $\context_\ext$ refer to nondeterministically chosen non-ports in the input graph. While every ordinary context node refers to exactly one node in the input graph, clonable context nodes refer to any number of nodes.

Throughout the rest of this paper, we shall continue to denote the components of an expansion operation $\ext$ by $V_\ext$, $E_\ext$, $\lab_\ext$, $\port_\ext$, $\dock_\ext$,  and $C_\ext$.

Applying $\ext$ to an argument graph $G=(V,E,\lab,\port)\in\graphs_\tau$ is possible if $\tau=|\dock|$. It then yields a graph of type $\tau'=|\port_\ext|$ by fusing the nodes in $\dock_\ext$ with those in $\port$. Moreover, all context nodes, that is, all nodes in the set $\context_\ext$, are fused with arbitrary (pairwise distinct) nodes in $V\setminus[\port]$ carrying the same label. Prior to this nodes in $C_\ext$ can be cloned an arbitrary number of times. Thus, the application of the operation to $G$ clones the nodes in $C$, fuses those in $\dock_\ext$ with those in $\port$, and fuses every node in $C_v$ ($v\in\context_\ext$) injectively with any node in $V\setminus[\port]$ carrying the same label. The port sequence of the resulting graph is $\port_\ext$.

Formally, let $\tau=|\dock_\ext|$ and $\tau'=|\port_\ext|$. Then $\ext$ is interpreted as the nondeterministic operation $\ext\colon\graphs_\tau\to\pow{\graphs_{\tau'}}$ defined as follows. Given a graph $G=(V,E,\lab,\port)\in\graphs_\tau$, a graph $H\in \graphs_{\tau'}$ is in $\ext(G)$ if it can be obtained by the following steps:
\begin{enumerate}
    \item Choose a graph $G'\in\clone{C_\ext}{\und\ext}$ and an isomorphism $\mu$ on $G'$ such that $\mu(G')=(V',E',\lab',\port')$ satisfies
    \begin{enumerate}
    \item $\mu(\New_\ext)\cap V=\emptyset$,
    \item $\mu(\dock_\ext)=\port$, and
    \item for all nodes $v\in\context_\ext$ and $v'\in\mu(C_v)$, it holds that $\mu(C_v)\subseteq V\setminus[\port]$ and $\lab_\ext(v)=\lab(v')$.
    \end{enumerate}
    \item Define $H=(V\cup V',E\cup E',\lab\funion\lab',\port')$.
\end{enumerate}
Note that, by the definition of $\funion$, the labels of nodes in $\dock_\ext$ are disregarded, i.e., the nodes which are ports in $G$ keep their labels as determined by $\lab$.\footnote{This idea has its origins in personal discussions between the second author and Berthold Hoffmann in~2019.} Thus, in specifying an expansion operation $\ext$ as above, the labels of nodes in $\dock_\ext$ can be dropped, essentially regarding these nodes as unlabelled ones.

In the following, we mainly view $\ext$ as an operation $\ext\colon\pow{\graphs_\tau}\to\pow{\graphs_{\tau'}}$ defined as $\ext(\G)=\bigcup_{G\in\G}\ext(G)$ for all $\G\subseteq\graphs_\tau$. We call $\und\ext=(V_\ext,E_\ext,\lab_\ext,\port_\ext)$ the \emph{underlying graph} of $\ext$, and may also denote $\ext$ as $(G,\dock,C)$ where $G=\und\ext$.

We now specialise expansion operations to extension operations by placing conditions on their structure. The intuition is to make sure that graphs are built bottom-up, i.e., $\ext$ always extends the input graph by placing nodes and edges ``on top'', and in such a way that all nodes of the argument graph are reachable from the ports. 
For this, we require that only new ports -- ports that do not correspond to ports of the input graph -- have outgoing edges, and that all ports of the input graph that are not ports of the resulting graph carry incoming edges. Noting that the former set is $\New_\ext$ and the latter is $[\dock_\ext]\setminus[\port_\ext]$, this results in the following pair of formal requirements:
\begin{enumerate}[label=(R\arabic*),leftmargin=*,series=R]
    \item\label{sources} $\{\src e\mid e\in E_\ext\}\subseteq\New_\ext$, and
    \item\label{forget} $[\dock_\ext]\setminus[\port_\ext]\subseteq\{\tar e\mid e\in E_\ext\}$
.
\end{enumerate}

\begin{example}
Figure~\ref{fig:extension_operation_and_graph} depicts an extension operation together with a graph to which it can be applied. 
In Figure~\ref{fig:resulting_graphs}, we see three different graphs, all resulting from the application of the extension operation in Figure~\ref{fig:extension_operation_and_graph} to the graph in Figure~\ref{fig:resulting_graphs}.
\end{example}

\begin{figure}
    \begin{center}
    \tikzset{
  every label/.style={draw=none, fill=none, inner sep=1pt}
}
     \begin{tikzpicture}[xscale=.5, yscale=.75]
     \tikzstyle{overtex}=[circle,line width=0.30mm,fill=white,draw,inner sep=.5pt, minimum size=2ex];
          
     \input{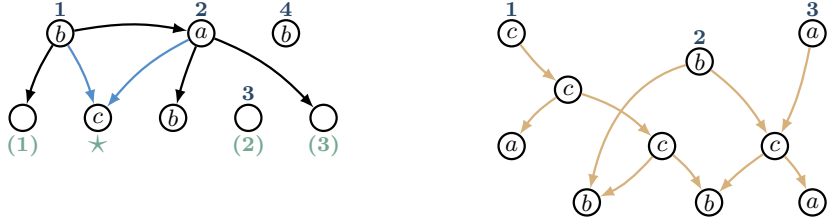}
     
     \begin{scope}[shift={(13,4)}]
     \node at (0,-2.5) [overtex, label=above:\prt 1] (Xz1) {\footnotesize $c$};
     \node at (8,-2.5) [overtex, label=above:\prt 3] (Xz2) {\footnotesize $a$};
     
     
     \node at (1.5,-3.5) [overtex] (Xu1) {\footnotesize $c$};
     \node at (5,-3) [overtex, label=above:\prt 2] (Xu2) {\footnotesize $b$};
     
     \node at (0,-4.5) [overtex] (Xv1) {\footnotesize $a$};
     \node at (4,-4.5) [overtex] (Xv2) {\footnotesize $c$};
     \node at (7,-4.5) [overtex] (Xv3) {\footnotesize $c$};
     
     \node at (2,-5.5) [overtex] (Xw1) {\footnotesize $b$};
     \node at (5.25,-5.5) [overtex] (Xw2) {\footnotesize $b$};
     \node at (8,-5.5) [overtex] (Xw3) {\footnotesize $a$};
     
    \draw[->,line width=0.30mm,-latex,UmUGold,bend right=7] (Xz1) to (Xu1);
    \draw[->,line width=0.30mm,-latex,UmUGold,bend left=10] (Xz2) to (Xv3);
    \draw[->,line width=0.30mm,-latex,UmUGold,bend right=10] (Xu1) to (Xv1);
    \draw[->,line width=0.30mm,-latex,UmUGold,bend left=10] (Xu1) to (Xv2);
    \draw[->,line width=0.30mm,-latex,UmUGold,bend right=30] (Xu2) to (Xw1);
    \draw[->,line width=0.30mm,-latex,UmUGold,bend left=10] (Xu2) to (Xv3);

    \draw[->,line width=0.30mm,-latex,UmUGold,bend left=10] (Xv2) to (Xw1);
    \draw[->,line width=0.30mm,-latex,UmUGold,bend left=10] (Xv2) to (Xw2);
    \draw[->,line width=0.30mm,-latex,UmUGold,bend right=5] (Xv3) to (Xw2);
    \draw[->,line width=0.30mm,-latex,UmUGold,bend left=10] (Xv3) to (Xw3);
    
    \path [use as bounding box] ;
    
    \end{scope}
    \end{tikzpicture}
    
    \end{center}
    \caption{The figure on the left shows an extension operation $\ext$ with three ports (indicated with numbers above the nodes), three docks (indicated with numbers below the nodes), and one clonable node (indicated with a star). The figure on the right shows a graph $G$ with three ports (again, indicated with numbers above the nodes).  The operation $\ext$ is applicable to $G$ because (i)  the number of ports of $G$ coincides with the number of docks of $\ext$, so the types are compatible, and (ii) the pair of nodes in $\context_\ext$ can be identified with nodes in $G$ outside of the ports of $G$ (in several ways).  }
    \label{fig:extension_operation_and_graph}
\end{figure}

\begin{figure}
    \begin{center}
    \tikzset{
  every label/.style={draw=none, fill=none, inner sep=1pt}
}
     \begin{tikzpicture}[xscale=.45, yscale=.75]
     \tikzstyle{overtex}=[circle,line width=0.30mm,fill=white,draw,inner sep=.5pt, minimum size=2ex];

     \node at (1,0) [overtex, label=above:\prt 1] (x1) {\footnotesize $b$};
     \node at (4.75,0) [overtex, label=above:\prt 2] (x2) {\footnotesize $a$};
     \node at (7,0) [overtex, label=above:\prt 4] (x3) {\footnotesize $b$};
     

     \node at (0,-2.5) [overtex] (Xz1) {\footnotesize $c$};
     \node at (8,-2.5) [overtex] (Xz2) {\footnotesize $a$};
     
     
     \node at (1.5,-3.5) [overtex] (Xu1) {\footnotesize $c$};
     \node at (5,-3) [overtex, label=above:\prt 3] (Xu2) {\footnotesize $b$};
     
     \node at (0,-4.5) [overtex] (Xv1) {\footnotesize $a$};
     \node at (4,-4.5) [overtex] (Xv2) {\footnotesize $c$};
     \node at (7,-4.5) [overtex] (Xv3) {\footnotesize $c$};
     
     \node at (2,-5.5) [overtex] (Xw1) {\footnotesize $b$};
     \node at (5.25,-5.5) [overtex] (Xw2) {\footnotesize $b$};
     \node at (8,-5.5) [overtex] (Xw3) {\footnotesize $a$};
     
    \draw[->,line width=0.30mm,-latex,bend left=7] (x1) to (x2);
    \draw[->,line width=0.30mm,-latex,bend right=10] (x1) to (Xz1);
    \draw[->,line width=0.30mm,-latex,bend left=35] (x2) to (Xw2);
    \draw[->,line width=0.30mm,-latex,bend left=20] (x2) to (Xz2);
     
    \draw[->,line width=0.30mm,-latex,UmUGold,bend right=7] (Xz1) to (Xu1);
    \draw[->,line width=0.30mm,-latex,UmUGold,bend left=10] (Xz2) to (Xv3);
    \draw[->,line width=0.30mm,-latex,UmUGold,bend right=10] (Xu1) to (Xv1);
    \draw[->,line width=0.30mm,-latex,UmUGold,bend left=10] (Xu1) to (Xv2);
    \draw[->,line width=0.30mm,-latex,UmUGold,bend right=30] (Xu2) to (Xw1);
    \draw[->,line width=0.30mm,-latex,UmUGold,bend left=10] (Xu2) to (Xv3);

    \draw[->,line width=0.30mm,-latex,UmUGold,bend left=10] (Xv2) to (Xw1);
    \draw[->,line width=0.30mm,-latex,UmUGold,bend left=10] (Xv2) to (Xw2);
    \draw[->,line width=0.30mm,-latex,UmUGold,bend right=5] (Xv3) to (Xw2);
    \draw[->,line width=0.30mm,-latex,UmUGold,bend left=10] (Xv3) to (Xw3);
    
    \path [use as bounding box] ;
    
    \end{tikzpicture}
    \qquad
    \tikzset{
  every label/.style={draw=none, fill=none, inner sep=1pt}
}
     \begin{tikzpicture}[xscale=.45, yscale=.75]
     \tikzstyle{overtex}=[circle,line width=0.30mm,fill=white,draw,inner sep=.5pt, minimum size=2ex];

     \node at (1,0) [overtex, label=above:\prt 1] (x1) {\footnotesize $b$};
     \node at (4.75,0) [overtex, label=above:\prt 2] (x2) {\footnotesize $a$};
     \node at (7,0) [overtex, label=above:\prt 4] (x3) {\footnotesize $b$};
     

     \node at (0,-2.5) [overtex] (Xz1) {\footnotesize $c$};
     \node at (8,-2.5) [overtex] (Xz2) {\footnotesize $a$};
     
     
     \node at (1.5,-3.5) [overtex] (Xu1) {\footnotesize $c$};
     \node at (5,-3) [overtex, label=above:\prt 3] (Xu2) {\footnotesize $b$};
     
     \node at (0,-4.5) [overtex] (Xv1) {\footnotesize $a$};
     \node at (4,-4.5) [overtex] (Xv2) {\footnotesize $c$};
     \node at (7,-4.5) [overtex] (Xv3) {\footnotesize $c$};
     
     \node at (2,-5.5) [overtex] (Xw1) {\footnotesize $b$};
     \node at (5.25,-5.5) [overtex] (Xw2) {\footnotesize $b$};
     \node at (8,-5.5) [overtex] (Xw3) {\footnotesize $a$};
     
    \draw[->,line width=0.30mm,-latex,bend left=7] (x1) to (x2);
    \draw[->,line width=0.30mm,-latex,bend right=10] (x1) to (Xz1);
    \draw[->,line width=0.30mm,-latex,bend left=10,color=NoUmUColour] (x1) to (Xv2);
    \draw[->,line width=0.30mm,-latex,bend left=35] (x2) to (Xw2);
    \draw[->,line width=0.30mm,-latex,bend right=10,color=NoUmUColour] (x2) to (Xv2);
    \draw[->,line width=0.30mm,-latex,bend left=20] (x2) to (Xz2);
     
    \draw[->,line width=0.30mm,-latex,UmUGold,bend right=7] (Xz1) to (Xu1);
    \draw[->,line width=0.30mm,-latex,UmUGold,bend left=10] (Xz2) to (Xv3);
    \draw[->,line width=0.30mm,-latex,UmUGold,bend right=10] (Xu1) to (Xv1);
    \draw[->,line width=0.30mm,-latex,UmUGold,bend left=10] (Xu1) to (Xv2);
    \draw[->,line width=0.30mm,-latex,UmUGold,bend right=30] (Xu2) to (Xw1);
    \draw[->,line width=0.30mm,-latex,UmUGold,bend left=10] (Xu2) to (Xv3);

    \draw[->,line width=0.30mm,-latex,UmUGold,bend left=10] (Xv2) to (Xw1);
    \draw[->,line width=0.30mm,-latex,UmUGold,bend left=10] (Xv2) to (Xw2);
    \draw[->,line width=0.30mm,-latex,UmUGold,bend right=5] (Xv3) to (Xw2);
    \draw[->,line width=0.30mm,-latex,UmUGold,bend left=10] (Xv3) to (Xw3);
    
    \path [use as bounding box] ;
    
    \end{tikzpicture}
    \qquad
    \tikzset{
  every label/.style={draw=none, fill=none, inner sep=1pt}
}
     \begin{tikzpicture}[xscale=.45, yscale=.75]
     \tikzstyle{overtex}=[circle,line width=0.30mm,fill=white,draw,inner sep=.5pt, minimum size=2ex];

     \node at (1,0) [overtex, label=above:\prt 1] (x1) {\footnotesize $b$};
     \node at (4.75,0) [overtex, label=above:\prt 2] (x2) {\footnotesize $a$};
     \node at (7,0) [overtex, label=above:\prt 4] (x3) {\footnotesize $b$};
     

     \node at (0,-2.5) [overtex] (Xz1) {\footnotesize $c$};
     \node at (8,-2.5) [overtex] (Xz2) {\footnotesize $a$};
     
     
     \node at (1.5,-3.5) [overtex] (Xu1) {\footnotesize $c$};
     \node at (5,-3) [overtex, label=above:\prt 3] (Xu2) {\footnotesize $b$};
     
     \node at (0,-4.5) [overtex] (Xv1) {\footnotesize $a$};
     \node at (4,-4.5) [overtex] (Xv2) {\footnotesize $c$};
     \node at (7,-4.5) [overtex] (Xv3) {\footnotesize $c$};
     
     \node at (2,-5.5) [overtex] (Xw1) {\footnotesize $b$};
     \node at (5.25,-5.5) [overtex] (Xw2) {\footnotesize $b$};
     \node at (8,-5.5) [overtex] (Xw3) {\footnotesize $a$};
     
    \draw[->,line width=0.30mm,-latex,bend left=7] (x1) to (x2);
    \draw[->,line width=0.30mm,-latex,bend right=10] (x1) to (Xz1);
    
    \draw[->,line width=0.30mm,-latex,bend left=10,color=NoUmUColour] (x1) to (Xu1);
    \draw[->,line width=0.30mm,-latex,bend left=30,color=NoUmUColour] (x1) to (Xv3);

    \draw[->,line width=0.30mm,-latex,bend right=20,color=NoUmUColour] (x2) to (Xu1);    
    \draw[->,line width=0.30mm,-latex,bend right=15] (x2) to (Xw1);
    \draw[->,line width=0.30mm,-latex,bend left=20] (x2) to (Xz2);
    \draw[->,line width=0.30mm,-latex,bend left=25,color=NoUmUColour] (x2) to (Xv3);
     
    \draw[->,line width=0.30mm,-latex,UmUGold,bend right=7] (Xz1) to (Xu1);
    \draw[->,line width=0.30mm,-latex,UmUGold,bend left=10] (Xz2) to (Xv3);
    \draw[->,line width=0.30mm,-latex,UmUGold,bend right=10] (Xu1) to (Xv1);
    \draw[->,line width=0.30mm,-latex,UmUGold,bend left=10] (Xu1) to (Xv2);
    \draw[->,line width=0.30mm,-latex,UmUGold,bend right=15] (Xu2) to (Xw1);
    \draw[->,line width=0.30mm,-latex,UmUGold,bend left=10] (Xu2) to (Xv3);

    \draw[->,line width=0.30mm,-latex,UmUGold,bend left=10] (Xv2) to (Xw1);
    \draw[->,line width=0.30mm,-latex,UmUGold,bend left=10] (Xv2) to (Xw2);
    \draw[->,line width=0.30mm,-latex,UmUGold,bend right=5] (Xv3) to (Xw2);
    \draw[->,line width=0.30mm,-latex,UmUGold,bend left=10] (Xv3) to (Xw3);
    
    \path [use as bounding box] ;
    
    \end{tikzpicture}
    \end{center}
    \caption{Three graphs in $\ext(G)$ where the extension operation $\ext$ and the graph $G$ are as in Figure~\ref{fig:extension_operation_and_graph}. The differences between the graphs reflect the number of times the nodes in $C_\ext$ have been replicated, and the choice of mapping from the nodes in $\context_\ext$ to nodes in $G$. The single node in $C_\ext$ has been replicated 0 times to obtain the first graph, once to obtain the second, and twice to obtain the third. In the first two graphs, the non-dock node labelled $b$ has been identified with one $b$-labelled node in $G$, and in the right graph with another.}
    \label{fig:resulting_graphs}
\end{figure}
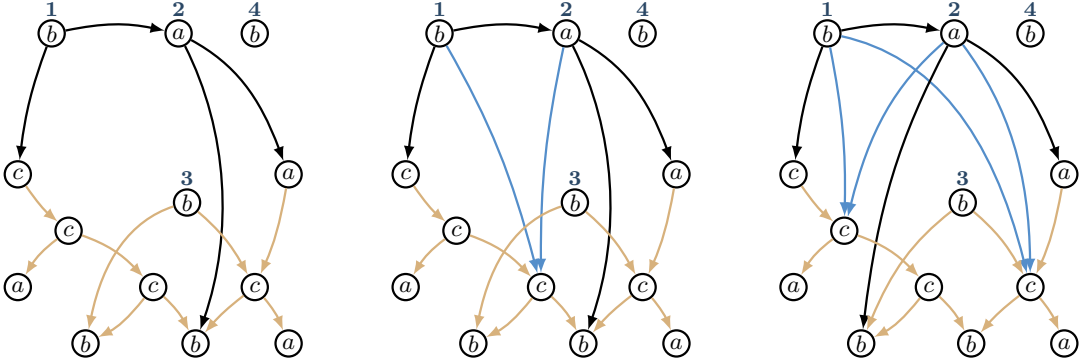

A \emph{graph extension algebra} is a $\sig$-algebra $\alg=((\pow{\graphs_\tau})_{\tau\in\types},(f_\alg)_{f\in\sig})$ where $\types$ is a finite subset of $\nat$ and every symbol in $\sig$ is interpreted as an extension or union operation or the set $\{\emptygraph\}$, where $\emptygraph$ is the empty graph $(\emptyset,\emptyset,\emptyset,\emptystr)$. Note that the operations of the algebra act on sets of graphs rather than on single graphs. This is necessary because of the nondeterministic nature of extension operations.

\begin{definition}[Graph Extension Grammar]
A (tree-based) \emph{graph extension grammar} is a pair $\Gamma=(g,\alg)$ where $\alg$ is a graph extension $\sig$-algebra for some signature $\sig$ and $g$ is a regular tree grammar over $\sig$. The graph language generated by $\Gamma$, denoted by $L(\Gamma)$, is defined as
\[
L(\Gamma)=\bigcup_{t\in L(g)}\val_\alg(t).
\]
\end{definition}

To simplify notation, we shall in the following assume that, in a graph extension grammar as above, $f=f_\alg$ for all $f\in\sig$, i.e., we use the operations themselves as symbols in $\sig$.

Before continuing on to parsing, let us pause to consider at a concrete example of graph extension grammar in the setting of natural language processing. 

\begin{myexample}

The graph extension grammar in Figure~\ref{fig:example_extension_operations} illustrates how the new formalism can model semantic representations.
In this example, the concepts \texttt{girl} and \texttt{boy} represent entities that can act as agents (which verbs cannot),
\texttt{try} and \texttt{persuade} represent verbs that require structural control, and \texttt{want} and \texttt{believe} represent verbs where this is not needed.

First, we take a top-down perspective to understand the tree generation.
The initial nonterminal is $S$. The base case for $S$ is the generation of an extension operation that creates a single node representing a \texttt{girl} or a \texttt{boy} concept. 
In all other cases, $S$ generates an extension operation that adds a verb and its outgoing edges, and in which the verb is the single port. 
The nonterminal $C$ has the same function as $S$ with the following two differences: it has no corresponding base case, and the ports of the extension operations mark both the agent of the verb (designated by an \texttt{arg0} edge) and the verb itself.
The nonterminal $S'$ can only generate extensions that create \texttt{girl} or \texttt{boy} nodes. As we shall see, this allows us to restrict the addition of \texttt{arg0} to persons.
Finally, there is a pair of nonterminals $U$ and $U'$ that both create union operations, the former with two resulting ports, and the latter with three.

Now, we take a bottom-up perspective to see how the extension operations of a tree are evaluated. 
Unless the tree consists of a single node, we will have several extensions generated by $S$ and $S'$ that create \texttt{girl} and \texttt{boy} nodes as the leaves of the tree. 
In this case, we can apply one or more union operations to concatenate their port sequences and make them visible to further operations.
The construction ensures that an \texttt{arg0} edge only can have a person, i.e., a valid agent, as its target. 
After the application of a union operation, it is possible to apply any extension that is applicable to $U$ (or $U'$), meaning that none of the non-port nodes are contextual nodes. 
The resulting graph can have one or two ports -- if the graph has one port, the applied extension operation was generated by $S$, and if the graph has two ports, the applied operation was generated by $C$.
In other words: $C$ signals that the graph is ready for the addition of a control verb, and $S$ that the graph is a valid semantic graph with one port. 
When sufficiently many nodes have been generated, it is no longer necessary (but still possible) to use the union operations.
Instead, we can add incoming edges to already generated nodes contextually (unless control is explicitly needed, as for control verbs).
The generated graphs can thus contain both structural and non-structural dependencies.
See Figure~\ref{fig:example_derivation_and_evaluation} for an example of a tree generated by the grammar in Figure~\ref{fig:example_extension_operations}, and its evaluation into a semantic graph. \hfill $\diamond$%
\input{fig-want-believe-extension-operations}%
\tikzset{
    node distance=1.4cm,
    every node/.style={font=\ttfamily\small,inner sep=2pt},
    terminal/.style={rounded corners, draw, fill=UmUGold!70, inner sep=4pt,font=\ttfamily},
    rule/.style={font=\small},nonterminal/.style={draw,inner sep=2pt, outer sep=0pt,execute at begin node=$,execute at end node=$,font=\ttfamily}
}
\tikzstyle{vertex}=[circle,draw,inner sep=2pt]
\tikzstyle{fvertex}=[circle,fill,draw,inner sep=2pt]
\tikzstyle{edge}=[draw,rectangle,line width=0.30mm,]
\tikzstyle{overtex}=[ellipse,line width=0.30mm,fill=white,draw,inner sep=.5pt, minimum size=2ex]

\begin{figure}[t!]
\hspace*{\fill}%
\scalebox{1.0}{
$\val_\alg\left(
\scalebox{0.7}{
  \begin{tikzpicture}[baseline=210]
  
     \node[circle=2,line width=0.30mm,draw,inner sep=0pt, minimum size=4cm, label=above left:{\Large $S$}] at (1.5,15.25) (c) {};
     \begin{scope}[shift={(0.4,14.65)}]
    \node at (1.2,1.75) [terminal, label=above:\prt 1] (want-x1) {want};
     
    \node at (0.0,0) [terminal] (want-y1) {boy};
    \node at (2.2,0) [overtex, label=below:\dk 1] (want-y2) {};
     
    \draw[->,line width=0.30mm,-latex,bend right=10] (want-x1) to node [pos=0.5, above, sloped] (want-a0) {\texttt{arg0}} (want-y1);
    \draw[->,line width=0.30mm,-latex,bend left=10] (want-x1) to node [pos=0.5, above, sloped] (want-a1) {\texttt{arg1}} (want-y2);
     \end{scope}
  
    \node[circle=2,line width=0.30mm,draw,inner sep=0pt, minimum size=5.5cm, label=above left:{\Large $S$}] at (1.5,10) (c0) {};
     \draw[line width=0.30mm] (c) to (c0);
    \begin{scope}[shift={(-0.5,9.4)}]
    \node[terminal] at (2,1.75) [label=above:\prt 1] (persuade-x1) {\texttt{persuade}};
     
    \node at (0.3,0) [terminal] (persuade-y1) {boy};
    \node at (2,0) [overtex, label=below:\dk 1] (persuade-y2) {};
    \node at (3.7,0) [overtex, label=below:\dk 2] (persuade-y3) {};
     
    \draw[->,line width=0.30mm,-latex,bend right=20] (persuade-x1) to node [pos=0.55, above, sloped] (a0) {\texttt{arg0}} (persuade-y1);
    \draw[->,line width=0.30mm,-latex,swap] (persuade-x1) to node [pos=0.5, above, sloped] (persuade-a1) {\texttt{arg1}} (persuade-y2);
    \draw[->,line width=0.30mm,-latex,bend left=20] (persuade-x1) to node [pos=0.6, above, sloped] (persuade-a1) {\texttt{arg2}} (persuade-y3);
    \end{scope}

     \node[circle=2,line width=0.30mm,draw,inner sep=0pt, minimum size=4cm, label=above left:{\Large $C$}] at (1.5,4.75) (c1) {};
     \draw[line width=0.30mm] (c0) to (c1);
     \begin{scope}[shift={(0.5,4.0)}]
    \node[terminal] at (1,1.75) [label=above:\prt 2] (believe-x1) {believe};
     
    \node at (-.1,0) [overtex, label=100:\prt 1, label=below:\dk 1] (believe-y1) {};
    \node at (2.1,0) [overtex, label=below:\dk 2] (believe-y2) {};
    \draw[->,line width=0.30mm,-latex,bend right=10] (believe-x1) to node [pos=0.5, above, sloped] (persuade-a0) {\texttt{arg0}} (believe-y1);
    \draw[->,line width=0.30mm,-latex,bend left=10] (believe-x1) to node [pos=0.5, above, sloped] (persuade-a1) {\texttt{arg1}} (believe-y2);
     \end{scope}
     
     \node[circle=2,line width=0.30mm,draw,inner sep=0pt, minimum size=1.5cm, label=above left:{\Large $U$}] at (1.5,1.5) (c2) {};
     \draw[line width=0.30mm] (c1) to (c2);
     \begin{scope}[shift={(1.5,1.5)}]
     \node[font=\sffamily\normalsize] at (0,0) (op-y1) {$\union11$};
     \end{scope}

     \node[circle=2,line width=0.30mm,draw,inner sep=0pt, minimum size=1.75cm, label=above:{\Large $S'$}] at (0,0) (c21) {};
     \draw[line width=0.30mm] (c2) to (c21);
     \begin{scope}[shift={(0,0)}]
     \node[terminal] at (0,0) [label=above:\prt 1] (girl-y1) {girl};
     \end{scope}

     \node[circle=2,line width=0.30mm,draw,inner sep=0pt, minimum size=1.75cm, label=above:{\Large $S$}] at (3,0) (c22) {};
     \draw[line width=0.30mm] (c2) to (c22);
     \begin{scope}[shift={(3,0)}]
     \node[terminal] at (0,0) [label=above:\prt 1] (boy-y1) {boy};
     \end{scope}

     \node[circle=2,line width=0.30mm,draw,inner sep=0pt, minimum size=.75cm] at (0,-1.75) (c211) {$\emptygraph$};
     \draw[line width=0.30mm] (c21) to (c211);

     \node[circle=2,line width=0.30mm,draw,inner sep=0pt, minimum size=.75cm] at (3,-1.75) (c221) {$\emptygraph$};
     \draw[line width=0.30mm] (c22) to (c221);

     \path[use as bounding box] (4.7,7);
     
     \draw[dotted,line width=0.30mm,bend right=15] (want-y1)
       .. controls +(1,-1) and +(-1,0) .. (2.5,13.5)
       .. controls +(2.8,0) and +(0,5) .. (4.5,7)
       .. controls +(0,-5) and +(2,2) ..
     (boy-y1);
     
     \draw[dotted,line width=0.30mm,bend right=15] (persuade-y1)
       .. controls +(1,-1) and +(-1,0) .. (2,8)
       .. controls +(2,0) and +(0,2) .. (4,4.5)
       .. controls +(0,-2) and +(1,2) ..
     (boy-y1);
  \end{tikzpicture}
} 
\right)\hspace{.8cm} = $
\hspace{.5cm}
\begin{tikzpicture}[baseline=0]

  \node[terminal, label=above:\prt 1] (try) at (2.5,1.75) {want};
  \node[terminal] (persuade) at (0,1.5) {persuade};
  \node[terminal] (believe) at (0,0) {believe};
  \node[terminal] (g) at (-1.5,-1.5) {girl};
  \node[terminal] (b) at (1.5,-1.5) {boy};
  \path[-latex] (try) edge[line width=0.20mm,bend left] node[pos=.5,right,outer sep=2pt] {arg0} (b) edge[line width=0.20mm,bend right] node[pos=.5,above,outer sep=2pt] {arg1} (persuade);
  \path[-latex] (persuade) 
  edge[line width=0.20mm, bend right=55] node[pos=.25,left,outer sep=2pt] {arg1} (g) 
  edge[line width=0.20mm] node[pos=.5,left,outer sep=2pt] {arg2} (believe) 
  edge[line width=0.20mm,bend left=40] node[pos=.25,right,outer sep=2pt] {arg0} (b);
  \path[-latex] (believe) 
  edge[line width=0.20mm,bend right] node[pos=.65,right,outer sep=2pt] {arg0} (g) 
  edge[line width=0.20mm,bend left] node[pos=.65,left,outer sep=2pt] {arg1} (b);
\end{tikzpicture}
  }
\hspace*{\fill}%

\caption{A (graphical representation of a) tree generated by the regular tree grammar in Figure~\ref{fig:example_extension_operations} and the graph resulting from its evaluation. The nonterminal generating each subtree is shown next to the respective circle. The identification of context nodes with previously created nodes is  indicated by dotted lines. In this example, there is only one possible candidate for a context node, hence the resulting graph is unique up to isomorphism.}
\label{fig:example_derivation_and_evaluation}
\end{figure}
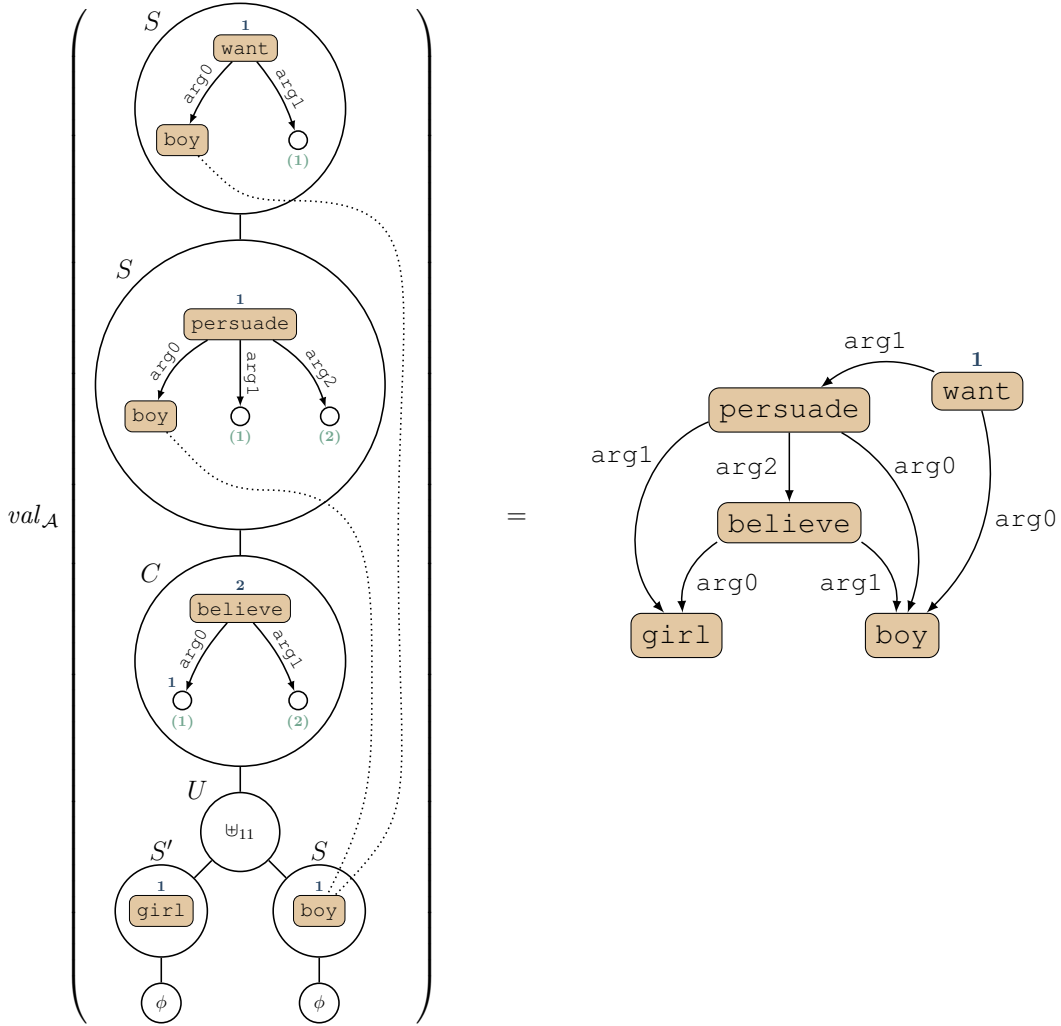%
\end{myexample}

Readers familiar with HRGs or Courcelle's hyperedge-replacement algebras \cite{Courcelle:91b} will have noticed that, when disregarding context nodes, productions using union and extension operations correspond to hyperedge replacement productions of two types:
\begin{enumerate}[label=(\alph*)]
\item $A\to B\union\tau{\tau'}C$ corresponds to a production with two hyperedges in the right-hand side, labelled~$B$ and~$C$, where the first one is attached to $\tau$ of the nodes to which the nonterminal in the left-hand side (which is labelled~$A$) is attached, and the second is attached to the remaining $\tau'$ nodes of the left-hand side. In particular, there are no further nodes in the right-hand side.
\item $A\to\ext[B]$ corresponds to a hyperedge-replacement production with a single nonterminal hyperedge labelled~$B$ in its right-hand side. Such productions are the ones responsible for actually generating new terminal items (nodes and edges).
\end{enumerate}
Generalised extension operations of arity greater than~$1$ can be constructed by defining so-called derived operations through composition of suitable extension and union operations. Given a set of such derived operations, the tree-to-tree mapping that replaces each occurrence of a derived operation by its definition is a tree homomorphism. As regular tree languages are closed under tree homomorphisms, this shows that the restriction to binary unions and unary extensions is no limitation (in contrast to requirements \ref{sources} and \ref{forget}, which ensure polynomial parsability).

\section{Parsing}

We now provide a basic polynomial parsing algorithm for graph extension grammars. For this, and throughout this section, let $\Gamma=(g,\alg)$ be a graph extension grammar, where $g=(N,\sig,P, S)$ and $\alg$ is a graph extension $\sig$-algebra. The goal is parsing, i.e., to decide the membership problem for $\Gamma$,  and in the positive case, produce a witness derivation in $\Gamma$. The membership problem is stated as follows:

\medskip\noindent
\begin{tabular}{@{\bfseries}ll@{}}
Input:& A graph $G$\\
Question:& Does it hold that $G\in L(\Gamma)$?
\end{tabular}

\medskip
By definition, extension operations keep the identities of nodes in the input graph unchanged, whereas new nodes added to the graph may be given arbitrary identities (as long as they are not in the input graph). For the union operation, it holds that renaming of nodes is necessary only if the node sets of the two argument graphs intersect. As a consequence, when evaluating a tree $t$, we can without loss of generality assume that operations never change node identities of argument graphs. In other words, if $G\in\val_\alg(t)$, then there is a (concrete) graph $G_\ad\in\val_\alg(t/\ad)$ such that $G_{\emptystr}=G$ and, for every $\ad\in\addr t$,
\begin{itemize}
\item if $t/\ad=\ext[t/\ad1]$, then $G_\ad\in\ext(G_{\ad1})$ and
\item if $t/\ad=t/\ad1\union\tau{\tau'}t/\ad2$, then $G_\ad=(V\cup V',E\cup E',\lab\funion\lab',\port\,\port')$ where $G_{\ad i}=(V_i,E_i,\lab_i,\port_i)$ for $i\in[2]$.
\end{itemize}
As a consequence, for every $\ad\in\addr t$, the graph $G_\ad$ is a subgraph of $G$.
More precisely, if $G=(V,E,\lab,\port)$ and $G_\ad=(V',E',\lab',\port')$ then $V'\subseteq V$, $E'\subseteq E$, and $\lab'=\lab|_{V'}$. 
Note that, while the graphs $G_\ad$ are usually not uniquely determined by $G$ and $t$, the important fact is that they do exist if and only if $G\in\val_\alg(t)$. 
We say in the following that the family $(G_\ad)_{\ad\in\addr t}$ is a \emph{concrete evaluation of $t$ into $G$}.
Hence, the membership problem amounts to decide, for the input graph $G$, whether there is a tree $t\in L(g)$ that permits a concrete evaluation into $G$.

The following lemma, which forms the basis of our parsing algorithm, shows that $G_\ad$ is determined by $\port'$ alone, as it is the subgraph of $G$ induced by the nodes that are reachable in $G$ from $\port'$ in $G$. More precisely, given a sequence $p\in\seq V$ of nodes of $G$, let $\reach Gp=(V',E|_{V'\times\edlab\times V'},\lab|_{V'},p)$, where $V'$ is the set of nodes reachable in $G$ by directed paths from (any of the nodes in) $\port$.

\begin{lemma}\label{le:reachability}
Let $(G_\ad)_{\ad\in\addr t}$ be a concrete evaluation of a tree $t\in\trees\sig$ into a graph $G$, and let $G_\ad=(V_\ad,E_\ad,\lab_\ad,\port_\ad)$ for every $\ad\in\addr t$. Then $G_\ad=\reach G{\port_\ad}$ for all $\ad\in\addr t$.
\end{lemma}

\begin{proof}
\newcommand{\cl}[1]{#1^\triangledown}%
Using the notations from the statement of the lemma, and furthermore denoting the set of all nodes that can be reached from $p\in\seq V$ in $G$ by $\cl p$, we first show the following claim:

\begin{claim}\label{claim:reachable contained}
For all $e\in E$ and $v\in V$, $\src e\in V_\ad$ implies $e\in E_\ad$, and $v\in\cl{\port_\ad}$ implies $v\in V_\ad$.
\end{claim}

We prove Claim~\ref{claim:reachable contained} by induction on the length of $\ad$. For $G_\ad=\emptystr$, the statement holds trivially. Now, assume that it holds for some $\ad\in\addr t$. Since the empty graph $\emptygraph$ has neither ports nor edges, two relevant cases remain.

\begin{firstcase}
$G_\ad$ is of the form $G_{\ad1}\union\tau{\tau'}G_{\ad2}$.
\end{firstcase}

Let $\{i,j\}=[2]$. We have to show that the statement holds for $G_{\ad i}$. Thus, assume that $\src e\in V_{\ad i}$ and $v\in\cl{\port_{\ad i}}$. Then $\src e\in V_\ad$ and thus, by the induction hypothesis, $e\in E_\ad$. By the definition of $\union\tau{\tau'}$ this means that $e\in E_{\ad i}$. Furthermore, $v\in\cl{\port_{\ad i}}$ implies $v\in\cl{\port_\ad}$ and thus, again by the induction hypothesis, $v\in V_\ad$. Hence, it remains to argue that no node in $G_{\ad j}$ can be reached from $\port_{\ad i}$. This follows readily from the just established fact that there is no edge $e'\in V\setminus V_{\ad i}$ for which $\src{e'}\in E_{\ad i}$ (together with the fact that $V_{\ad i}\cap V_{\ad j}=\emptyset$).

\begin{case}
$G_\ad$ is of the form $\ext(G_{\ad1})$ for an extension operation $\ext$.
\end{case}

To show that the statement holds for $G_{\ad1}$, let $G_\ad=(V_{\ad1}\cup V', E_{\ad1}\cup E', \lab_{\ad1}\funion\lab', \port')$, where $(V',E',\lab',\port')$ is obtained from $\und\ext$ as in the definition of $\ext(G_{\ad1})$. By the induction hypothesis, $\src e\in V_\ad$ implies $e\in E_\ad$. Thus, $\src e\in V_{\ad1}$ implies $e\in E_{\ad1}$ unless $e\in E'$. However, by requirement~\ref{sources}, all $e\in E'$ satisfy $\src e\in[\port']\setminus[\port]$, which is equivalent to $\src e\notin V_{\ad1}$ because $V_\ad\setminus V_{\ad1}=[\port']\setminus[\port]$. This shows that $e\in E_{\ad1}$.

Now, let $v\in\cl{\port_{\ad1}}$. If $v\notin V_{\ad1}$, consider the first edge $e$ on a path from a node in $[\port_{\ad1}]$ to $v$ such that $\tar e\notin V_{\ad1}$. Then $\src e\in V_{\ad1}$ but $e\notin E_{\ad1}$, contradicting the previously established fact that $\src e\in V_{\ad1}$ implies $e\notin E_{\ad1}$. This finishes the proof of Claim~\ref{claim:reachable contained}.

The other direction of the second part of Claim~\ref{claim:reachable contained} remains to be shown:

\begin{claim}\label{claim:reachable contains}
$V_\ad\subseteq\cl{\port_\ad}$.
\end{claim}

This time, we proceed by induction on the size of $t/\ad$. The statement is trivially true for $G_\ad=\emptygraph$. Thus, as before, there are two cases to distinguish.

\begin{firstcase}
$G_\ad$ is of the form $G_{\ad1}\union\tau{\tau'}G_{\ad2}$.
\end{firstcase}

By the induction hypothesis, Claim~\ref{claim:reachable contains} holds for $G_{\ad1}$ and $G_{\ad2}$. Consequently, $V_\ad=V_{\ad1}\cup V_{\ad2}\subseteq\cl{\port_{\ad1}}\cup\cl{\port_{\ad2}}=\cl{\port_{\ad}}$.

\begin{case}
$G_\ad$ is of the form $\ext(G_{\ad1})$ for an extension operation $\ext$.
\end{case}

Again, let $G_\ad=(V_{\ad1}\cup V', E_{\ad1}\cup E', \lab_{\ad1}\funion\lab', \port')$, where $(V',E',\lab',\port')$ is obtained from $\und\ext$ as in the definition of $\ext(G_{\ad1})$. (In particular, $\port'=\port_\ad$.) By the induction hypothesis, $V_{\ad1}\subseteq\cl{\port_{\ad1}}$. Moreover, by requirement~\ref{forget}, for every node $v\in[\port_{\ad1}]$ it either holds that $v\in[\port_\ad]$, or there are $u\in[\port_\ad]$ and $e\in E'$ with $\src e=u$ and $\tar e=v$. Hence, $[\port_{\ad1}]\subseteq\cl{\port_\ad}$ and thus $V_{\ad1}\subseteq\cl{\port_{\ad1}}\subseteq\cl{\port_\ad}$. Since, furthermore, $V'\setminus V_{\ad1}\subseteq[\port_\ad]$, this shows that $V_\ad\subseteq [\port_\ad]$, and thus $V_\ad\subseteq\cl{\port_\ad}$, as claimed.
\end{proof}

We note the following immediate consequence of Lemma~\ref{le:reachability}:

\begin{corollary}\label{co:empty graph}
Let $(G_\ad)_{\ad\in\addr t}$ be a concrete evaluation of a tree $t\in\trees\sig$ into a graph $G$. If $G_\ad=\reach G{\emptystr}$ for some $\ad\in\addr t$, then $G_\ad=\emptygraph$.
\end{corollary}

It follows from Corollary~\ref{co:empty graph} that a nonterminal $A$ of type $\emptystr$ can only derive trees that evaluate to the empty graph $\emptygraph$. 
We can therefore replace all rules where $A$ appears on the left-hand side by a single rule that takes $A$ to the empty graph, without changing the generated graph language. 
The construction suggests the following normal form for graph extension grammars.

\begin{lemma}[normal form]\label{le:normal form}
$\Gamma$ can be transformed into a graph extension grammar $\Gamma'=(g',\alg')$ with $L(\Gamma')=L(\Gamma)$ such that the co-domain $s$ of every extension and union operation in $\alg'$ satisfies $|s|\ge 1$.
\end{lemma}

\begin{proof}
We may assume that $L_A(g)\neq\emptyset$ for all $A\in N$. By Corollary~\ref{co:empty graph}, $\sigma(A)=\emptystr$ then implies that $\val_\alg(L_A(g))=\{\emptygraph\}$. It follows that all rules with the left-hand side $A$ can be replaced by a single rule $A\to\emptygraph$.
\end{proof}

From here on, let us assume that $\Gamma$ is in the normal form given by Lemma~\ref{le:normal form}. Turning our attention back on Lemma~\ref{le:reachability}, we find the beginnings of a recursive parsing algorithm: In the following, consider an input graph $G=(V,E,\lab,\port)$, a nonterminal $A\in N$, and a sequence $p\in\seq V$ with $\ndlab(p)=\sigma(A)$. To decide whether $\reach Gp\in\val_\alg(L_A(g))$ (and thus whether $G\in L(\Gamma)$ if $A=S$ and $p=\port_G$), there are two cases. If $p=\emptystr$, then $\reach Gp=\emptygraph\in\val_\alg(L_A(g))$ if and only if the production $A\to\emptygraph$ is in $P$. Otherwise, we need to check each production with the left-hand side $A$, as follows:
\begin{firstcase}
If the production is of the form $A\to B_1\union{\tau_1}{\tau_2}B_2$, then we recursively need to determine whether $\reach G{p_i}\in\val_\alg(L_{B_i}(g))$ for $i=1,2$, where $p=p_1p_2$ and $|p_1|=\tau_1$.
\end{firstcase}
\begin{case}\label{parsing ext}
If the production is of the form $A\to\ext(B)$, the basic intuition is that we need to check all possible mappings of the edges in $E_\ext$ to edges in $E$ that have their sources in $p$. However, since nodes in $C_\ext$ can be cloned, and since some context nodes of $\ext$ may be isolated, the technicalities are slightly more involved.

To simplify the reasoning, let us first assume that $\und\ext$ does not contain isolated context nodes, i.e., that it satisfies the following stronger variant of~\ref{forget}:

\begin{enumerate}[resume*=R]
\item\label{forget2} $V_\ext\setminus[\port_\ext]\subseteq\{\tar e\mid e\in E_\ext\}$.
\end{enumerate}

Let $\port_\ext=u_1\cdots u_k$, $p=v_1\cdots v_k$, and define $\New=\{v_i\mid\text{$i\in[k]$ and $u_i\notin[\dock_\ext]$}\}$.
In other words, $\New$ contains the nodes in $p$ which correspond to those in $\New_\ext$, which are the only nodes that may have edges to other nodes in $\ext$. Now, let $E_P=\{e\in E\mid\src e\in\New\}$ and $V_p=[p]\cup\{\tar e\mid e\in E_p\}$. Note that, by~\ref{forget2}, $V_p$ contains all the nodes of $V$ that must be put into correspondence with nodes in $V_\ext$, while the nodes in in $V \setminus V_p$ are to be left for future, recursive, invocations of the parsing algorithm.  We say that a surjective mapping $m\colon V_p\to V_\ext$ is a \emph{matching of $\ext$ to $\reach Gp$} if it:
\begin{enumerate}[label=(M\arabic*),leftmargin=*,series=M]
\item\label{node label preservation} preserves node labels, that is, $\lab_\ext(m(v))=\lab(v)$ for all $v\in V$ with $m(v)\notin[\dock_\ext]$,
\item\label{clone only clonable} clones only clonable nodes, meaning that $|m^{-1}(u)|=1$ for all $u\in V_\ext\setminus C_\ext$, and
\item\label{edge preservation} bidirectionally preserves edges, that is, $(v,\ell,v')\in E_p \iff (m(v),\ell,m(v'))\in E_\ext$ for all $v,v'\in V_p$ and all $\ell\in\ndlab$.
\end{enumerate}

Now, by Lemma~\ref{le:reachability}, the parsing algorithm needs to check recursively whether there is a matching of $\ext$ to $\reach Gp$ such that $\reach G{m^{-1}(\dock_\ext)}\in\val_\alg(L_{B}(g))$.\footnote{Recall that $[\dock_\ext]\cap C_\ext=\emptyset$, and thus $m^{-1}(\dock_\ext)\in V^{|\dock_\ext|}$.}

Finally, to avoid the stronger assumption~\ref{forget2}, we have to handle any isolated context nodes $u\in V_\ext$. Since, by injectivity, such a node $u$ must be mapped to any node $v$ in $\reach Gp$ with the same label, except those in $V_p$. Thus, we only need to check the existence of such nodes outside $V_p$. Note that clonable isolated context nodes $u$ can (and must) be disregarded because they can be cloned zero times, thus having no effect. In summary, let
\[
\#_\ext(\ell)=|\{u\in \context_\ext\setminus C_\ext\mid\lab_\ext(u)=\ell\}|
\]
for all $\ell\in\ndlab$. Then a matching $m$ must, in addition, satisfy the following condition:

\begin{enumerate}[resume*=M]
\item\label{isolated} For all $\ell\in\ndlab$, $\reach Gp$ contains at least $\#_\ext(\ell)$ nodes $v$ with $\lab(v)=\ell$ that are not in $V_p$.
\end{enumerate}

Note that this condition is independent of $m$.
\end{case}

The pseudocode of the algorithm is shown in Algorithm~\ref{alg:parsing}.
\begin{algorithm}[t]
\small
\algrenewcommand\algorithmicprocedure{\textbf{procedure}}
\caption{Parse a graph $G$ with respect to a graph extension grammar $\Gamma=(g,\alg)$\label{alg:parsing}}
\begin{algorithmic}[1]
\Procedure{Parse}{$G=(V,E,\lab,\port)$}
  \State \Return \Call{Parse\_rec}{$S,\port$} \Comment{Invoke recursive procedure with initial nonterminal $S$}
\EndProcedure
\Procedure{Parse\_rec}{$A\in N$, $p\in\seq V$}
  \If {$\sigma(A)=\emptystr$} \Comment{Leverage normal form}
    \State \Return $(A\to\emptygraph)\in P$
  \ElsIf {$\mathit{result}[A,p]$ defined} \Comment{Use memoïsed result}
    \State \Return $\mathit{result}[A,p]$
  \EndIf
  \State memoïse $\mathit{result}[A,p]\gets\mathit{false}$ \Comment{No success so far}
  \ForAll {$(A\to B_1\union{\tau_1}{\tau_2}B_2)\in P$}\label{outer1}
    \State let $p=p_1p_2$ where $|p_1|=\tau_1$
    \If {\Call{Parse\_rec}{$B_1,p_1$} and \Call{Parse\_rec}{$B_2,p_2$}}
      \State memoïse $\mathit{result}[A,p]\gets\mathit{true}$ \Comment{Memoïse success}
      \State \Return $\mathit{true}$
    \EndIf
  \EndFor
  \ForAll {$(A\to\ext[B])\in P$}\label{outer2}
    \ForAll {matchings $m$ of $\ext$ to $\reach Gp$}\label{for loop}
      \If {\Call{Parse\_rec}{$B,m^{-1}(\dock_\ext)$}}
        \State memoïse $\mathit{result}[A,p]\gets\mathit{true}$ \Comment{Memoïse success}
        \State \Return $\mathit{true}$
      \EndIf
    \EndFor
  \EndFor
  \State \Return $\mathit{false}$
\EndProcedure
\end{algorithmic}
\end{algorithm}%
As we shall see, the algorithm runs in polynomial time in the size of the input graph. However, the maximum length of the type sequences used appears in the exponent, so the degree of the polynomial depends on a property of the grammar. In the following theorem, we let $|G|=|V|+|E|$ for a graph $G=(V,E,\lab,\port)$.

\begin{theorem}\label{th:main}
Let $\sig$ be $\types$-typed and $c=\max\{\tau\mid \tau\in\types\}$. If applied to an input graph $G$, Algorithm~\ref{alg:parsing} terminates in $O(|G|^{2c+1})$ steps and returns $\mathit{true}$ if and only $G\in L(\Gamma)$.
\end{theorem}

\begin{proof}
Let $n=|G|$. Correctness should be clear from Lemma~\ref{le:reachability}, the assumption that $\Gamma$ is in normal form (cf.~Lemma~\ref{le:normal form}), and the preceding discussion. Further, by the choice of $c$, there are at most $O(n^c)$ distinct parameters \textsc{Parse\_rec} can be called with. Since we consider a fixed graph extension grammar $\Gamma$, this is a polynomial. For the same reason, the two outer \textbf{for} loops of the algorithm are executed a constant number of times during each call. It remains to consider the loop starting on line~\ref{for loop}. Undoubtedly, the naïve implementation that enumerates all matchings is exponential. However, it is clear from the loop body that we only really need to determine all sequences $d_1\cdots d_k=m^{-1}(\delta_1\cdots\delta_k)$, where $\delta_1\cdots\delta_k=\dock_\ext$ and $m$ is a matching of $\ext$ to $\reach Gp$. (Recall that $m^{-1}$, restricted to $[\dock_\ext]$, is a function by \ref{clone only clonable}.) Since there are no more than $O(n^c)$ such sequences it remains to be shown that, for each of them, it is possible to check in linear time whether it can be extended to a matching. It bears repeating that we, in the paragraphs that follow, are only looking to establish the existence of such an extension, not compute a full set of complete matchings, as that would be too costly.

Clearly, condition~\ref{isolated}, which is independent of the choice of the sequence $d_1\dots d_k$, can be checked in linear time. Thus, it remains to be shown how to verify that the mapping $d_i\mapsto\delta_i$ can be extended to a matching under the assumption that \ref{isolated} is satisfied.

For this, let us use the notation of Case~\ref{parsing ext} in the discussion preceding the theorem (i.e., the definition of matchings following \ref{forget2}). In particular, $E_P=\{e\in E\mid\src e\in\New\}$ and $V_p=[p]\cup\{\tar e\mid e\in E_p\}$. Consider a candidate sequence $d_1\cdots d_k\in V_p^k$. We have to check whether there is a matching $m$ such that $m(d_1\cdots d_k)=\delta_1\cdots\delta_k$. For this, let
\[
\lab'(u)=\left\{\begin{array}{@{}ll@{}}
\lab_\ext(u)&\text{if $u\in V_\ext\setminus[\dock_\ext]$}\\
\lab(d_i)&\text{if $u=\delta_i$ for some $i\in[k]$.}
\end{array}\right.
\]

The idea is now to compute a characterisation, or \emph{profile}, of the nodes in $V_\ext$ and $V_p$ that allows us to quickly decide whether such a matching $m$ exists. In short, we record for each node~$u$  a set of triples, where each triple describes an incoming edge to $u$ that must be covered by the matching. For every node $u\in V_\ext$, let $\profile_\ext(u)=\{(i,\ell,\lab'(u))\mid(u_i,\ell,u)\in E_\ext\}$. Similarly, for $v\in V_p$, let $\profile_G(v)=\{(i,\ell,\lab(v))\mid(v_i,\ell,v)\in E\}$. Figure~\ref{fig:profile} illustrates the construction with an example.

We now combine the profiles for the nodes in $V_\ext$ into a single, composite, profile for all of $\ext$, giving us a set of sets of edge descriptions. For each such edge description, we count how many of its type we need to find. This can be an exact number, or a lower bound, depending on whether some of the target nodes are clonable and can be left out or replicated as desired.
Thus, define $\profiles_\ext=\{\profile_\ext(u)\mid u\in V_\ext\setminus[\port_\ext]\}$ and let, for all $\pi\in\profiles_\ext$,
\[
\mult(\pi)=\left\{\begin{array}{@{}ll@{}}
\equal m & \nexists\, u\in  C_\ext\colon\profile_\ext(u)=\pi\\
\greater m & \text{otherwise}
\end{array}\right.
\]
where $m=|\{u\in V_\ext\setminus([\port_\ext]\cup C_\ext)\mid\profile_\ext(u)=\pi\}|$. Then, by the definition of matchings, there exists a matching $m\colon V_p\to V_\ext$  such that $m(d_1\cdots d_k)=\delta_1\cdots\delta_k$ if and only if
\begin{itemize}
\item the subgraph of $G$ induced by $[p]$ is isomorphic to the subgraph of $\und\ext$ induced by $[\port_\ext]$,
\item for all $i\in[k]$, $\profile_\ext(\delta_i)=\profile_G(d_i)$,\footnote{In fact, because of the previous item this only needs to be checked for $\delta_i\notin[\port_\ext]$.} and
\item for all $\pi\in\profiles_\ext$,
\begin{align*}
|\{v\in V_p\mid\profile(v)=\pi\}| &=k\quad \text{if $\mult(\pi)=\equal k$ for some $k\in\nat$}\\
|\{v\in V_p\mid\profile(v)=\pi\}| &\ge k\quad \text{if $\mult(\pi)=\greater k$ for some $k\in\nat$}\enspace.
\end{align*}
\end{itemize}

It should be clear that these conditions can be checked in linear time.
\end{proof}

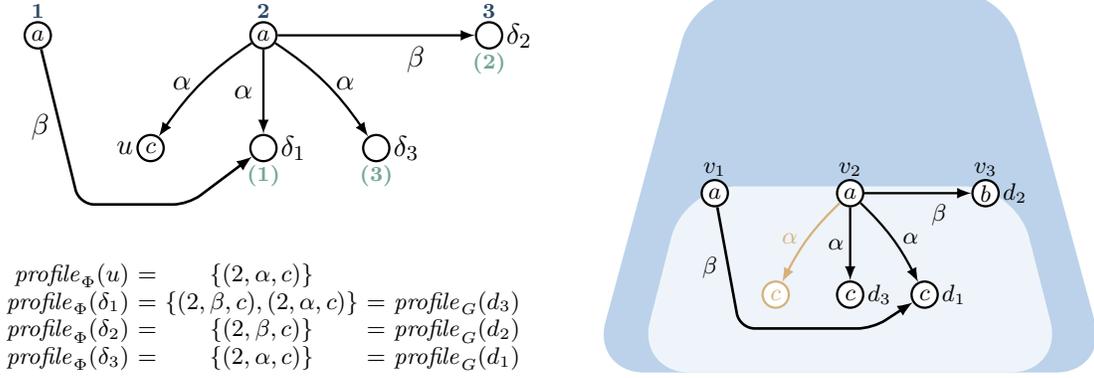
\begin{figure}
    \centering
    \tikzset{
  every label/.style={draw=none, fill=none, inner sep=1pt}
}
 \tikzset{
    sh2n/.style={shift={(0,1)}},
    sh2s/.style={shift={(0,-1)}},
    sh2e/.style={shift={(1,0)}},
    sh2w/.style={shift={(-1,0)}},
    sh2nw/.style={shift={(-1,1)}},
    sh2ne/.style={shift={(1,1)}},
    sh2sw/.style={shift={(-1,-1)}},
    sh2se/.style={shift={(1,-1)}},
      rc/.style={rounded corners=2mm,line width=1pt},
  }

\centering
     \begin{tikzpicture}[xscale=.5, yscale=.75]
     \tikzstyle{overtex}=[circle,line width=0.30mm,fill=white,draw,inner sep=.5pt, minimum size=2ex];
    
        \node at (0.5, 5.5) (lab) {\Large $\ext$:};
    
     \begin{scope}[shift={(0,1.5)}]      
     \node at (0.0,2) [overtex, label=above:\prt 1] (x1) {\footnotesize $a$};
     \node at (6.0,2) [overtex, label=above:\prt 2] (x2) {\footnotesize $a$};
     \node at (12.0,2) [overtex, label=right: $\delta_2$, label=above:\prt 3, label=below:\dk 2] (x3) {};
     
     \node at (3.0,0.0) [overtex, label=left: $u$] (y1) {\footnotesize $c$};
     \node at (6.0,0.0) [overtex, label=right: $\delta_1$, label=below:\dk 1] (y2) {};
     \node at (9.0,0.0) [overtex, label=right: $\delta_3$, label=below:\dk 3] (y3) {};
     
    \node at (2.5, 3) (d) {};
    \node at (9.5, 3) (e) {};
    \draw[->,line width=0.30mm,-latex] (x2) to node[below, pos=.7] {$\beta$} (x3) ;

    \draw[->,line width=0.30mm,-latex,bend right=10] (x2) to node[left, pos=0.5] {$\alpha$}  (y1);
    \draw[->,line width=0.30mm,-latex]  (x2) to node[left, pos=0.5] {$\alpha$} (y2) ;
    \draw[->,line width=0.30mm,-latex,bend left=10] (x2) to node[right, pos=0.5] {$\alpha$}  (y3);
    
    \node at (1.1, -1) (b) {};
    \node at (4.0, -1) (c) {};
    \draw[->,line width=0.30mm,-latex,rc] (x1) -- node[left, pos=0.5] {$\beta$} (b.center) --  (c.center) -- (y2);
    \end{scope}
    
    \node at (6, -1.5) (text) {\footnotesize
        $\begin{array}{r@{\ }c@{\ }c@{\ }c@{\ }l}
        \profile_\ext(u)& = &\{(2,\alpha,c)\}\\    
        \profile_\ext(\delta_1)& = &\{(2,\beta,c),(2,\alpha,c)\}&=&\profile_G(d_3)\\    
        \profile_\ext(\delta_2)& = & \{(2,\beta,c)\} &=&\profile_G(d_2)\\
        \profile_\ext(\delta_3)& = &\{(2,\alpha,c)\}&=&\profile_G(d_1)
        \end{array}$};


    \begin{scope}[shift={(18,-.5)},scale=0.6]
    
        \node at (0.5, 7) (ogx1) {};
    \node at (6.0, 7) (ogx2) {};
    \node at (11.5, 7) (ogx3) {};
    \node at (-4.0, -2) (ogy1) {};
    \node at (6.0, -2) (ogy2) {};
    \node at (16.0, -2) (ogy3) {};

\fill[NoUmUColour!40,rounded corners=7mm]  ($(ogx1.north west)+(-.9,.7)$) -- ($(ogx3.north east)+(.9,.7)$) -- ($(ogy3.south east)+(1,-1)$) -- ($(ogy1.south west)+(-1,-1)$) -- cycle;

    \node at (0.0, 1.2) (ngx1) {};
    \node at (6.0, 1.2) (ngx2) {};
    \node at (12.0, 1.2) (ngx3) {};
    \node at (-2.0, -2) (ngy1) {};
    \node at (6.0, -2) (ngy2) {};
    \node at (14.0, -2) (ngy3) {};
\fill[NoUmUColour!10,rounded corners=7mm]  ($(ngx1.north west)+(-.9,.7)$) -- ($(ngx3.north east)+(.9,.7)$) -- ($(ngy3.south east)+(1,-1)$) -- ($(ngy1.south west)+(-1,-1)$) -- cycle;

    \node at (0.0,2) [overtex, label=above:{\footnotesize $v_1$}] (x1) {\footnotesize $a$};
     \node at (6.0,2) [overtex, label=above:{\footnotesize $v_2$}] (x2) {\footnotesize $a$};
     \node at (12.0,2) [overtex, label=above:{\footnotesize $v_3$},label=right:{\footnotesize $d_2$}] (x3) {\footnotesize $b$};
     
     \node at (2.7,-1) [overtex,draw=UmUGold] (y1) {\textcolor{UmUGold}{\footnotesize $c$}};
     \node at (6.0,-1) [overtex,label=right:{\footnotesize $d_3$}] (y2) {\footnotesize $c$};
     \node at (9.3,-1) [overtex,label=right:{\footnotesize $d_1$}] (y3) {\footnotesize $c$};
     
    \node at (2.5, 3) (d) {};
    \node at (9.5, 3) (e) {};
    \draw[->,line width=0.30mm,-latex] (x2) to node[below, pos=.7] {\footnotesize $\beta$} (x3);
    
    \draw[color=UmUGold,->,line width=0.30mm,-latex,bend right=10] (x2) to node[left, pos=0.5] {\textcolor{UmUGold}{\footnotesize $\alpha$}}  (y1);
    \draw[->,line width=0.30mm,-latex]  (x2) to node[inner sep=2pt, left, pos=0.5] {\footnotesize $\alpha$} (y2) ;
    \draw[->,line width=0.30mm,-latex,bend left=10] (x2) to node[right, pos=0.5] {\footnotesize $\alpha$}  (y3);
    
    \node at (1.1, -2) (b) {};
    \node at (7.0, -2) (c) {};
    \draw[->,line width=0.30mm,-latex,rc] (x1) --  node[left, pos=0.5] {\footnotesize $\beta$} (b.center) -- (c.center) -- (y3);

  \end{scope}

    \end{tikzpicture}

    \caption{On the top-left is a schematic view of an extension operation $\ext$ with three ports and equally many docks. Deciding whether there is a matching $m$ such that $m(d_i)=\delta_i$ for $i\in[3]$ boils down to computing the functions $\profile_\ext$ and $\profile_G$ for all nodes in $V_\ext$ and $V_p$, respectively. (Note that $\profile_\ext(\delta_i)$ depends on $d_i$.) Based on the result, we find that $\reach Gp$ must contain the part drawn in brown for $d_1d_2d_3$ to extend to a matching, but only the existence of such an edge and node is relevant. }
    \label{fig:profile}
\end{figure}

While Algorithm~\ref{alg:parsing} only returns $\mathit{true}$ or $\mathit{false}$, i.e., solves the pure membership problem, we can instead solve the parsing problem by changing the memoïsation to store for every nonterminal $A$ and every $p$ a tree that gives rise to the generated subgraph $\reach Gp$ (or all of those trees).

We note here that the theoretical upper bound $O(|G|^{2c+1})$ proved above is unlikely to be reached in the average practical case, because an efficient implementation would only test candidate sequences $d_1\cdots d_k$ such that $\profile_\ext(\delta_i)=\profile_G(d_i)$ for all $i\in[k]$. 
In fact, there are many useful special cases for which better bounds can be proved. 
Two examples are mentioned in the next theorem, in which we use the following notation similar to the profiles used in the proof above, but restricted to the first and second components: For every node $u\in V_\ext$, we let $\profile'_\ext(u)=\{(i,\ell)\mid(u_i,\ell,u)\in E_\ext\}$ and $\profiles'_\ext=\{\profile'_\ext(u)\mid u\in V_\ext\setminus[\port_\ext]\}$.

\begin{theorem}
Let $\sig$ be $\types$-typed and $c=\max\{\tau\mid \tau\in\types\}$. 
\begin{enumerate}
\item Assume that, for every extension operation $\ext$ of $\alg$ and all $u\in[\dock_\ext]\setminus[\port_\ext]$ and $u'\in C_\ext$, it holds that $\profile'_\ext(u)\neq\profile'_\ext(u')$. Then Algorithm~\ref{alg:parsing} runs in time $O(|G|^{c+1})$.
\item Assume that the following hold:
\begin{enumerate}
\item For every extension operation $\ext$ of $\alg$ and all distinct $u\in[\dock_\ext]\setminus[\port_\ext]$ and $u'\in V_\ext\setminus[\port_\ext]$, we have $\profile'_\ext(u)\neq\profile'_\ext(u')$.
\item For all distinct productions of the form $A\to\ext[B]$ and $A\to\ext'[C]$, where $\ext$ and $\ext'=(V'_\ext,E'_\ext,\lab'_\ext,\dock'_\ext,C'_\ext)$, it holds that
\[
\{\profile'_\ext(u)\mid u\in V_\ext\setminus C_\ext\}\not\subseteq\profiles'_{\ext'}
\text{\quad or\quad}
\{\profile'_{\ext'}(u)\mid u\in V_{\ext'}\setminus C_{\ext'}\}\not\subseteq\profiles'_\ext\enspace.
\]
\item For all $A\in N$, if there is a production of the form $A\to B\union\tau{\tau'}C$ in $P$, then there is no other production with the left-hand side $A$.
\end{enumerate}
Then Algorithm~\ref{alg:parsing} runs uniformly in linear time.
\end{enumerate}
\end{theorem}

\begin{proof}
The assumption of the first statement implies that only a constant number of sequences $d_1\cdots d_k$ in the proof of Theorem~\ref{th:main} can fulfil the requirement $\profile_\ext(\delta_i)=\profile_G(d_i)$ for all $i\in[k]$. Hence, the \textbf{for} loop starting on line~\ref{for loop} can be implemented to run in linear time, which proves the statement.

For the second statement, note that the first assumption, which strengthens the assumption of the first part of the theorem, implies that at most one candidate sequence $d_1\cdots d_k$ can fulfil the requirement that $\profile_\ext(\delta_i)=\profile_G(d_i)$ for all $i\in[k]$. Moreover, by the second assumption, if there is such a sequence for a production $A\to\ext[B]$ considered in the loop starting on line~\ref{outer2}, then all other productions to be considered in that loop lack such a sequence. Finally, the third assumption makes sure that at most one production is to be considered on line~\ref{outer1}, and if there is one, then there is no production of the form required on line~\ref{outer2}.

It follows that the algorithm will decompose $G$ in a deterministic top-down process, each step taking constant time. This proves the assertion.
\end{proof}

\section{Conclusion}

We have introduced a graph grammar formalism that is both a restriction and an extension of hyperedge-replacement graph grammars, namely graph extension grammars.
These grammars are formalised in a tree-based fashion, consisting of an algebra over graphs and a regular tree grammar that generates expressions over the operations of the algebra. 
The graphs are directed graphs with ports, constructed using operations of two kinds: 
The first is disjoint union, and the second is a family of unary operations that add new nodes and edges to an existing graph~$G$. The augmentation is done in such a way that all new edges lead from a new node to either
\begin{enumerate}[label=(\roman*)]
\item a port in~$G$,
\item an arbitrary node in~$G$, chosen only by its node label, or
\item any number of such nodes.
\end{enumerate}
Graph extension grammars are essentially a restriction of contextual hyperedge-replacement grammars~\cite{drewes:2012contextual}. 
By way of example, we have shown that they can model both the structural and non-structural reentrancies that are common in semantic graphs of formalisms such as AMR. 
We have provided a parsing algorithm for our formalism and proved it to be correct and have a polynomial-time running time. Finally, we have formulated conditions on the grammar that, if met, make parsing particularly efficient.

There are several promising directions for future work. 
On the theoretical side, we would like to apply the new ideas presented here to the formalisms by Groschwitz et al.~\cite{groschwitz-etal-2017-constrained} and Björklund et al.~\cite{bjorklund:2016} to see if also these can be made to accommodate contextual rules without sacrificing parsing efficiency. 
It is also natural to generalise graph extension grammars to string-to-graph or tree-to-graph transducers to facilitate translation from natural-language sentences or dependency trees to AMR graphs. 
On the empirical side, we are interested in algorithms for inferring extension and union operations from AMR corpora, and in training neural networks to translate between sentences and AMR graphs using trees over graph extension algebras as an intermediate representation. 
Such efforts would make the new formalism available to current data-driven approaches in NLP, with the aim of adding structure and interpretability to machine-learning workflows.

\bibliographystyle{plain}
\bibliography{main}

\end{document}